\documentclass[journal,10pt]{IEEEtran}

\usepackage{cite}
\usepackage{subfigure}
\usepackage{epsfig}
\usepackage{amsfonts}
\usepackage{amsmath}
\usepackage{amsthm}
\usepackage{amssymb}
\usepackage{algorithm}
\usepackage{algorithmic}
\usepackage{multirow}
\usepackage{graphicx}
\usepackage{subfigure}
\usepackage{color}

\newcommand{\abs}[1]{\lvert #1 \rvert}

\newcommand{\card}[1]{\abs{#1}}

\newcommand{\nth}[1]{{#1}^{\text{th}}}
\newcommand{\smfrac}[2]{{\textstyle{\frac{#1}{#2}}}}

\newtheorem{lemma}{Lemma}[section]

\newtheorem{definition}{Definition}

\newtheorem{theorem}{Theorem}

\newtheorem{remark}{Remark}

\newcommand{\HRule}{\vspace{0.05in} \hrule \vspace{0.05in}}

\newcommand{\ie}{{\it i.e.\/}}

\newcommand{\mw}{\mathsf{MW}}
\newcommand{\Ex}{{\mathbb{E}}}
\newcommand{\mh}{{\mathcal{H}}}
\newcommand{\bO}{{\mathcal{O}}}
\newcommand{\mV}{{\mathcal{V}}}
\newcommand{\mE}{{\mathcal{E}}}

\newcommand{\changecolor}[1]{{\textcolor{black}{#1}}}

\allowdisplaybreaks

\date{}

\vspace{-0.2in}
\begin{document}

\title{Low Delay MAC Scheduling for Frequency-agile Multi-radio Wireless Networks}
\author{Avhishek Chatterjee, Supratim Deb, Kanthi Nagaraj, Vikram Srinivasan
\thanks{
A. Chatterjee is with The University of Texas at Austin (e-mail: avhishek@utexas.edu), 
S. Deb and V. Srinivasan are with Alcatel-Lucent 
(e-mail: \{supratim.deb, vikram.srinivasan\}@alcatel-lucent.com), 
and Kanthi Nagaraj is with Stanford University (e-mail: kanthi.cn@stanford.edu).
}
\thanks{Manuscript received Apr 13, 2012; revised date	Sep 1, 2012}	 
}
\maketitle
\begin{abstract}
Recent trends suggest that cognitive radio based wireless networks will be
frequency agile and the nodes will be equipped with multiple radios capable of tuning
across large swaths of spectrum.  The MAC scheduling problem in such networks refers
to making intelligent decisions on which communication links to activate at which time instant and over
which frequency band. The challenge in designing a low-complexity distributed MAC,
that achieves low delay, is posed by two additional dimensions of cognitive radio networks: 
interference graphs and data
rates that are frequency-band dependent, and explosion in number of feasible
schedules due to large number of available frequency-bands. In this paper, we propose
MAXIMAL-GAIN MAC, a distributed MAC scheduler for frequency agile multi-band networks that
simultaneously achieves the following: {\em (i) optimal network-delay scaling} with respect to
the  number of communicating pairs, {\em (ii) low computational complexity} of
$\bO(\log^2(${maximum degree of the interference
graphs}$))$ which is independent of the number of frequency bands, number of radios
per node, and overall size of the network, 
and {\em (iii) robustness}, i.e., it can be adapted to
a scenario where nodes are not synchronized and control packets could be lost.  
Our proposed MAC also achieves a throughput provably within a constant fraction (under
isotropic propagation) of the maximum throughput.  Due to a recent impossibility
result, optimal delay-scaling could only be achieved with some  amount of throughput
loss~\cite{shahtradeoff09}. Extensive
simulations using OMNeT++ network simulator shows that, compared to a multi-band
extension of a state-of-art CSMA algorithm (namely, Q-CSMA), our asynchronous
algorithm achieves a $2.5\times$ reduction in delay while achieving at least 85\% of the
maximum achievable throughput. Our MAC algorithms are derived from a
novel {\em local search} based technique.

\end{abstract}


\vspace{-0.15in}
\section{Introduction}
\label{sec:intro}


There are two emerging trends in the evolution of cognitive radio based
wireless networks. First, new regulations allow the possibility for a single technology to
utilize spectrum across very diverse range.  For example, FCC mandated the use of
unutilized TV spectrum in the $50\--698$~MHz band for unlicensed
access~\cite{FCC-NBP}. Along with unlicensed spectrum in 2.4~GHz and 5.1~GHz,
this implies that we now have {\em fragments} of unlicensed spectrum available across
several GHz. Second, advances in RF technology allows cognitive radio based wireless devices 
to tune across large range of spectrum. For example, a recent product by Radio Technology
Systems~\cite{rts}  makes it possible for radios to tune their center frequencies
from 100~MHz to 8~GHz. Such radios are referred to as {\em frequency agile} radios.
Note that, at any given time, the radio can tune to a single center frequency with a
upper limit on the operating bandwidth (typically 40~MHz).  These two trends imply
that wireless devices have greater flexibility in adapting to the network.
In this paper, we identify and address new challenges in MAC design that
arise as a consequence of these trends.


We consider a wireless network where each network node is equipped with a cognitive radio
that can detect the presence of a large number of available frequency bands.
\changecolor{Since our work is on distributed MAC design, we implicitly assume that 
MAC has knowledge of available frequency bands; this knowledge could come from a
layer/module that connects to a database or from a layer/module that connects to sensing
equipments.} The availability of frequency bands is quasi-static, i.e., the available
frequency bands can change roughly at the time-scale of session durations and not at the
time-scale of milliseconds. In the terminology of cognitive radio networks, the wireless
network nodes are {\em secondary transmitters} and we assume that {\em primary
transmitters} vacate the channels for large time-scales of the order of session
durations\footnote{Wireless networks over TV whitepsaces follow such a model.}.  We also
assume that each wireless node has potentially multiple transmitting radios to make full
use of the cognitive capabilities.  Each wireless node is allowed to transmit over all or
a large subset of the detected frequency bands\footnote{A node can transmit over a
frequency bands either if the band is unlicensed or if the node is deployed by an entity
having license to operate over the band.}. For example, in a given location, the network
codes can avail all unlicensed spectrum that could have 5 frequency bands consisting of
3~non-overlapping channels  in 2.4~GHz ISM along with 2~TV whitespace bands
$512\--524$~MHz and $692\--698$~MHz\footnote{We use the terms {\em frequency bands} and
{\em channels} interchangeably.}.

\changecolor{Thus, from a MAC design point of view, two aspects of cognitive radio capabilities
of nodes are most relevant: firstly, multi-band RF capability of the radio nodes,
and secondly, ability of a radio to switch frequency band with minimal switching overheads
of less than {\em millisecond}~\cite{rts}.} While these provide additional flexibility for
more efficient sharing of radio resources; however this also gives rise to new challenges
in MAC scheduling.  The key MAC scheduling question in such a network is: {\em at every
instant, which set of communicating node-pairs should operate over which frequency band
and using which radio?} Compared to traditional wireless networks, the MAC scheduling
problem in cognitive radio based wireless networks is more complex due to two primary
reasons:

\begin{enumerate}

\item Due to a potentially large number of available frequency bands, the number of feasible schedules
increases considerably. Indeed, the number of available frequency bands can easily be few
tens (e.g., $3$~non-overlapping frequency bands in 2.4~GHz, $13$~in 5~GHz, and
another $5\--15$ frequency bands in TV whitespaces) in today's networks. Thus,
choosing the ``right" schedule becomes computationally more challenging.

\item Since the available frequency bands can have diverse propagation
characteristics, the interfering
neighbors and the data rates depend on the operating frequency band. This is because wireless
path loss inversely depends on the square of the operating frequency.  This is unlike
traditional wireless networks that operate either over a single frequency band or
over multiple frequency bands with homogeneous propagation properties (referred to as
multi-channel networks in the literature). 

\end{enumerate}

We wish to design practical and low-complexity MAC schedulers while addressing these
unique challenges.


{\bf Design goal:} 
While the MAC scheduling question can be answered
differently depending on suitable performance goals, our primary
objective is to minimize network queueing delay since network traffic is getting dominated by
delay sensitive applications. Thus, we wish to design a MAC that meets the
following simultaneous objectives: {\em (i)} low delay,
i.e., total queueing delay scales linearly with the number of communication links which is the optimal scaling
{\em (ii)} low computational complexity and protocol overhead, {\em (iv)} and very
importantly, robust in the sense that it does not require nodes to be synchronized
and allows for losses in control overhead.  The reader might wonder why we do not
mention maximizing network throughput as an objective. This is because, it has been
shown in~\cite{shahtradeoff09} that there is a fundamental trade-off between
achieving 100\% throughput, achieving delay that scales polynomially with the network
size, and polynomial complexity of schedule computation, in the sense that all three
cannot be achieved together. Thus, if we wish to design a low complexity scheduler
that minimizes network delay, we must sacrifice on throughput. In this paper, we
design a MAC scheduler that achieves our three objectives with minimal throughput
loss.

\begin{sloppypar}
{\bf Our approach:} The theoretical underpinning of our work comes from so
called {\em local search} based approaches for complex optimization problems.
Local search is an iterative method, where, in each step, the current solution
is improved by looking at the {\em neighborhood} of current solution.  The
choice of {\em neighborhood} of a feasible solution, along with how the
transitions happen (it could be greedy, or random, or based on some transition
probability structure) from one solution to another, determine the
computational complexity and performance guarantee of the algorithm. Local
search based algorithms are attractive due to their simplicity and amenability
to distributed implementation. One popular local search based technique that
has been applied to wireless scheduling problems is {\em Glauber Dynamics}
based scheduling schemes and this is shown to achieve 100\%
throughput~\cite{srikant_qcsma,alohaworks,jiawal09}. Referred to as Q-CSMA,
these schemes are distributed and are computationally very light in each
iteration ($\bO(1)$ computation time per iteration). However, convergence time
of a Glauber Dyamics based scheme could be exponential in the size of the
graph~\cite{Dyer98oncounting} which could adversely impact the network delay.
Indeed, not much is known about delay guarantees of Glauber Dynamics based
scheduling algorithms, except possibly for very special classes of network
graphs (see Section~\ref{sec:literature}) in a single band scenario.  The
convergence issues of Glauber dynamics could exacerbate in multi-band cognitive-radio based
wireless networks as the number of feasible schedules grows with the number of bands and
radios. In this paper, {\em we propose an alternative local search based
distributed scheduling algorithm characterized by a very different
solution-neighborhood and transition structure}.  Our scheme achieves optimal
delay scaling in multi-band wireless networks but with at most $\bO(1)$
fraction of throughput loss under isotropic propagation. {\em Since it is impossible
to achieve 100\% throughput and optimal delay scaling
simultaneously~\cite{shahtradeoff09}, our approach is complementary to Glauber
dynamics based scheme; our scheme optimizes delay with some minimal loss in throughput, whereas,
Glauber dynamics based scheme optimizes throughput and sacrifice on delay guarantees.} 
 Our scheme is distributed, requires $\bO(\ln^2\Delta)$ ($\Delta$ is
maximum degree of interference graph) computation time per schedule, and can be
adapted to asynchronous setting.

Since MAC scheduling is a link-level decision, for simplicity and ease of exposition,
we derive our scheme assuming all network traffic to be single-hop. We also describe
later in the paper how all our results can be easily extended to a multi-hop network
setting using a standard back-pressure based approach.

\end{sloppypar}

\vspace{-0.1in}
\subsection{Our Contributions}

Our main contributions are as follows:

\begin{sloppypar}
\begin{list}{}{\itemsep=0pt \parskip=0pt \parsep=0pt \topsep=0pt \leftmargin=0.18in}

\item[1.]{{\em Design of low-delay MAC scheduler for cognitive radio networks:}} 
First, for a synchronous network, we
design a distributed scheduling (called the {\sf MAXIMAL-GAIN} algorithm) algorithm
that provably achieves all of the following:
{\em (i)} an average network-wide total queue length that scales as $HM$ where $H$ is the number of
transmit-receive pairs and $M$ is the number of available frequency bands,
{\em (ii)} computational overhead $\bO(\log^2 \Delta)$ times the time required to
exchange RTS-CTS message between two neighboring nodes (here $\Delta$ is the maximum
degree of the network graph), and
{\em (ii)} any throughput within a factor $\beta$ of the throughput region where
$\beta$ is a topology dependent parameter that is $\bO(1)$ for practical networks like
grid, hexagonal deployment, random deployment with isotropic propagation etc.

\item[2.]{{\em Developing extensions for asynchronous network:}} In an
asynchronous cognitive-radio based wireless network, we identify several issues that can severely
impair the performance of scheduling algorithms developed for synchronous setting.
Inspired by the design philosophy of the synchronous MAXIMAL-GAIN algorithm, we
design a robust MAXIMAL-GAIN like CSMA based algorithm for asynchronous settings.

\item[3.]{{\em Evaluation:}} We provide detailed evaluation of our algorithm. To evaluate
the benefits of designing a delay-centric MAC and to understand the throughput-loss
of our MAC, we compared our algorithm with a multi-band multi-radio adaptation of a  
distributed and practically implementable MAC that is known to achieve 100\% throughput (namely,
Q-CSMA~\cite{srikant_qcsma}). We compared the asynchronous version of our algorithm
with Q-CSMA which is synchronous in nature.  We report results of extensive
simulations over the OMNet network simulator. We show that our MAXIMAL-GAIN algorithm
achieves $2\--3\times$ reduction in delay, while achieving over $85\%$ of the
maximum possible throughput.

\end{list}
\end{sloppypar}

\vspace{-0.2in}
\section{Related Work}
\label{sec:literature}

The extensive research on MAC scheduling in single band networks
(see~\cite{linshrsri06} for an excellent survey) can be broadly divided into two
classes: max-weight computation based and Glauber dynamics based.  The first class of
approach is inspired by the seminal work~\cite{taseph92} which proves that
maximum-weight (MW) scheduling achieves 100\% throughput. The MW scheduler MAC
activates at every instant a {\em non-interfering} set of links such that the total
of weighted data-rates of the activated links is maximized, where, weight of a link
is roughly defined by the number of backlogged packets.  Furthermore, a recent
work~\cite{modiano_delay} (also see~\cite{maximal_delay_neely,neely-mw-delay09}) has
shown that an MW schedule achieves order optimal delay scaling with the number of links. 
However, computing MW schedule is NP-hard, leading
to considerable research on approximate MW schedules~\cite{tasRand,vaidya08,sarkar06,
shroff}.  However, extending these works to multi-band wireless networks, with
distributed implementation and low complexity, is not easy.  As discussed in
Section~\ref{sec:intro}, the second class of
approach~\cite{srikant_qcsma,alohaworks,jiawal09} is motivated by so called Glauber
dynamics based local search. These scheme achieve 100\% throughput but do not provide
delay guarantees in general networks.

Some examples of work that focus on low delay algorithms in single band networks
are~\cite{junsha07,jagsha08}.  
For single band networks, low delay CSMA based
algorithms for special classes of interference graphs have been proposed:
\cite{delayoptshah} proposes a delay-optimal CSMA based MAC for {\em polynomial
growth} networks (i.e., the number of nodes $r$ hops away from a node is polynomial
in $r$), and~\cite{sriwal11} proposes polynomial delay Q-CSMA based algorithm for
bounded degree networks. 

Multi-radio multi-channel resource allocation is addressed in \cite{compl_mrmc,
thyaga-multi, vaidya08}. 
Unlike our work, none of these
works account for the fact that different bands can have different propagation
characteristics. 

The inspiration of our MAC comes from application of {\em local search} based
techniques to develop approximation algorithms for NP-hard optimization
problems. We refer the reader to~\cite{ls_survey} for an excellent survey. 

\vspace{-0.15in}
\section{Preliminaries}
\label{sec:setting}

\vspace{-0.1in}
\subsection{Network Model}



We consider a wireless network with $\mV$ as the set of nodes.  The radio
resources available are multiple frequency bands with diverse propagation
characteristics. {\em By frequency band, we mean a contiguous slice of spectrum
of width at most $B_{max}$, the maximum tunable range of a radio. If there
is a larger contiguous band $W$ available, we split it into  ``bands" of equal
width $B_{max}$, with the last band having width possibly less than $B_{max}$}.
Since the frequency bands have diverse propagation, the data rate between a
transmitter-receiver pair in the network and the interfering neighbors of the
transmitter-receiver are frequency dependent. The
availability of frequency bands do not change for the duration of a session. 
In the terminology of cognitive radio
networks, the wireless network nodes are {\em secondary transmitters} and we assume
that {\em primary transmitters} vacate the channels for large time-scales of the
order of session durations\footnote{Wireless networks over TV whitepsaces can be
modeled this way.}.

Each network node has cognitive radio capabilities by which we mean the following:

\begin{enumerate}

\item Each wireless node has $K$ half-duplex radios capable of tuning across a
large frequency range (e.g., from 100~MHz to 2.5~GHz, \cite{cafaro-motorola}).
Each node can detect and communicate over all the frequency bands. This model
subsumes the case where a node can only detect/use a subset of the available
bands by setting the PHY data rates (we will elaborate on this later in the
section) over the prohibited frequency bands to be zero. We also remark that
our work easily generalizes to the case where different mesh nodes have
different number of radios.

\item We assume that at any given time any radio can tune to
only a single frequency band (e.g., 500MHz, 2.4GHz etc.) and the bandwidth can
range from 0 to $B_{max}$. 

\end{enumerate}

We consider all traffic between one-hop communication pairs which is a standard
approach to designing a MAC scheduler. Since MAC scheduling is a link level
decision, the single hop setting is not limiting at all and the extension to a
multi-hop setting can be done in a standard manner~\cite{taseph92} using queue
back-pressure based approach (see Section~\ref{sec:ext} for the required
modifications).  Towards this end, we will denote by $\mh$ the set of
single-hop node pairs and let $\card{\mh}=H$. Since we are considering a
multi-band wireless network, we call a node pair $(u,v)$ single hop if $u$ can
transmit to (or receive from) $v$ at non-zero rate over at least one of the
available bands. We consider a slotted system.  Associated with each single hop
node pair $h\in \mh$ is a stochastic arrival process $\{A_h(t)\}$, where
$A_h(t)$ is the number of packets arriving to node pair $h$ in time-slot $t$.
Let $\lambda_h=E[A_h(t)]$ be the average arrival rate at node pair $h$. We will
also assume that the arrival processes have bounded second moments.  We will
assume that the arrival process is {\em i.i.d.} across time slots. Let $D_h(t)$
be the number of packets that depart from node pair $h$ in time-slot $t$ and
let $q_h(t)$ be the queue length at the end of time-slot $t$.  We will assume
that the arrivals happen at the beginning of the slot. The queue length
evolution can be described by $$q_h(t+1)=[q_h(t)+A_h(t)-D_h(t)]^+\ .$$ We will
also refer to $\Lambda$ as the set of all possible arrival rate vectors that
can keep the queues bounded for some sequence $\{D_h(t)\}_t$ that is feasible.
The arrivals are also assumed to have bounded second moments.

Transmission over a node-pair can happen at different rates depending upon the
frequency band it is operating on. We next extend the standard definition of
link (and associated link data rate) between any node pair to account for the
frequency band dependent data rate.
\subsubsection*{Network graph and generalized link}
We assume that the spectrum available in the system is fragmented and spans a
large range (e.g., 50-700MHz, 2.4GHz and 5.2GHz). Let $M$ be the number of
frequency bands.  We will also assume that the bands are ordered in the real axis from
left to right, and we will simply say band~$f_j,\ j=1,2,\hdots,M$ to mean the
$\nth{j}$ frequency band. We will also use ``band-$j$" and ``frequency band
$f_j$" interchangeably.

Radio propagation physics dictates that, if all other parameters remain same,
the received signal strength at a receiver is inversely proportional to the
square of the carrier frequency \cite{Rappaport} (i.e., halving the
frequency doubles the received signal strength). Thus, it is possible
that two nodes in the network may be able to communicate on a frequency
band $f_1$ but unable to communicate on frequency $f_2$.

The network connectivity in frequency band-$j$ is modeled as a graph
$G_j(\mV,{\mE}_j),j=1,2,\hdots,M$ where $\mV$ is the set of network nodes and for two nodes
$u$ and $v$, $(u,v)\in {\mE}_j$ if $u$ can communicate with $v$ over frequency band $f_j$
at a non-zero data rate. Also, we will denote by $\Delta_j$ the maximum degree in the graph
$G_j(\mV,{\mE}_j)$ and $\Delta=\max_j\Delta_j\ .$\\

\vspace{-0.15in}
\begin{definition}[Generalized link]
A generalized link is defined
by {\em (i)} the band $m$ over which the link exists, and {\em
(ii)} the ordered node-pair $(i,j)$ to indicate that node~$i$ is the
transmitting node of the link and $j$ is the receiving node of the link (more
precisely, one of the radios of node~$i$ ($j$) is the transmitting (receiving)). We will also use the notation $\text{band}(l)=m,\
\text{head}(l)=j,\ \text{tail}(l)=i$ to describe link-$l$.
\end{definition}

\vspace{-0.15in}
\changecolor{
\begin{remark}
Our algorithms use the notion of generalized link that abstracts out diverse
propagation in different bands. Indeed, there could be multiple generalized
links between two nodes each corresponding to different bands with different
data rates and different interference neighborhoods.  
\end{remark}
}

In the rest of the paper, we say {\em link} to mean
generalized link. Note that, multiple links can exist between two nodes if they
can communicate at non-zero rate over multiple frequency bands. We 
denote the data rate over generalized link-$l$ by $r_l$. We define the weight of a 
generalized link $l$, $w_l$, as follows.

\vspace{-0.05in}
\begin{definition}[Weight of a link]
We define the weight of link $l$ as $w_l = q_h$, if link $l$ is on hop
$h \in \mh$, i.e., $w_l$ is the queue backlog on hop $h$.
\end{definition}

\vspace{-0.1in}
\subsubsection*{Interference model}
Our interference model is the widely used {\bf secondary
interference}~\cite{shamazshr06,chakarsar05,linshrsri06} model.
In this model, two links $l_1,l_2$ interfere over a frequency band $f_j$ if any
one of the end points of $l_1$ can decode messages from any one of the end points of
$l_1$ even at the lowest possible modulation and coding.  Such an interference model is
essential for the functioning of RTS/CTS based CSMA algorithm. 

Thus, each generalized link-$l$ has a set of interfering links denoted by $I_l$.
Thus, if $l'\in I_l$, $l$ and $l'$ cannot transmit simultaneously. Note that the notion of
generalized links easily accounts for the fact that the interfering node-pairs of a
node-pair can be different in different bands. To keep the exposition focused, for now
and unless stated otherwise, we ignore {\em adjacent channel interference} (ACI)
(due to  imperfect RF hardware causing transmit power to leak into adjacent bands) 
in multi-band networks~\cite{ourdtvpap09}. In Section~\ref{sec:ext}, we outline
how our model, algorithm, and analysis can be easily adapted to account for ACI.

Some of the useful notations are shown in Table~\ref{tab:notations}.
\begin{table}[tbh]
\footnotesize
\caption{\label{tab:notations}
{Table of Notations}}
\begin{center}
{
\begin{tabular}{|l|l|}
\hline
$M$ & Number of available frequency band \\ \hline
$f_j$ & Frequency band indexed $j,\ j=1,2,\hdots, M$ \\ \hline
$\mV$ & Set of mesh nodes; $\card{\mV}=V$ \\ \hline
${\mathcal{E}}_j$ & Set of  connected node pairs in $f_j$\\ \hline
${\mathcal{H}}$ & Set of single-hop \\
                & source destination pairs; $\card{H}=H$\\ \hline
$\mathcal{L}$ & Set of links, $\card{\mathcal{L}}=L$ \\ \hline
head($l$) and & Receiving and transmitting \\
tail($l$) &  node of link~$l$. \\ \hline
band($l$) & Frequency band of link~$l$. \\ \hline
$r_l$ & Data rate of generalized link $l$ \\ \hline
$I_l$ & Set of links that interfere with link~$l$ \\
      & (the interfering links belong to band($l$).) \\ \hline
$w_l$ & Weight of generalized link $l$ ($w_l=q_{h}$ \\
        & if $l$ is a link between $h$ node-pair)\\ \hline
$S_t$ & Set of activated links under \\
      & {\sf MAXIMAL GAIN} schedule at time-slot $t$ \\ \hline
$S_t(f_j)$ & Set of activated links over band $f_j$ \\
      & in {\sf MAXIMAL GAIN} schedule at time-slot $t$\\ \hline
\end{tabular}
}
\end{center}

\end{table}
\vspace{-0.1in}
\subsubsection*{Some useful definitions}

\begin{sloppypar}

Finally, we provide a few useful definitions. Using standard definition, we say
$I\subseteq \mV$ is independent set of a graph $G(\mV,\mE)$ if no two vertices
in $I$ share an edge.

\vspace{-0.05in}
\begin{definition}[2-hop neighborhood MIS cover] The
2-hop neighborhood maximum independent set (MIS) cover of a vertex $v$
(denoted by $\kappa_j(v)$) in a network graph $G_j(\mV,{\mE}_j)$ over
frequency band $f_j$ is defined as the maximum size of an independent set in the
graph induced by the nodes that lie within 2~hops of $v$ in $G_j(\mV,{\mE}_j)$.
We also define 2-hop neighborhood MIS cover of $G_j(\mV,{\mE}_j)$ as
$\kappa_j=\max_{v\in V}\kappa_j(v).$
\end{definition}

Finally, we will use the standard definitions of interference degree
in~\cite{chakarsar05,linshrsri06}, throughput optimality and stability
region~\cite{linshrsri06}. Denote by $\mathcal{S}$ the set of all feasible
schedules.

\vspace{-0.05in}
\begin{definition}[Interference degree] The interference degree of link $l$ is
defined as the maximum number of links that can be activated simultaneously over
band($l$) in a schedule. We will denote the
interference degree of link $l$ as $\beta_l$ and we also define the interference
degree of the network as $\beta_{max}=\max_l\beta_l$.
\end{definition}


\vspace{-0.05in}
\begin{sloppypar}
\begin{definition} [Stability region and throughput optimality] 
\label{def:to} 
The stability
region $\Lambda$ is the set of all arrival rates $\lambda$ such that $\lambda =
\sum_i \phi_i S_i, S_i \in \mathcal{S}$ for some ${\phi_i}$, with $1\geq \phi_i \geq
0, \forall i$ and $\sum_i \phi_i=1$.

An algorithm which stabilizes (the queues remain bounded) all arrival rates $\lambda
\in \Lambda$ is said to be throughput-optimal. An algorithm which stabilizes all
arrival rates $\lambda \in \frac{\Lambda}{\mu}$, for some $\mu >1$ is said to be a
$\mu$-factor throughput-optimal algorithm.

\end{definition}

\end{sloppypar}
\end{sloppypar}


\vspace{-0.15in}
\subsection{Scheduling Problem}
\label{sec:probstate}

The scheduling problem we consider is to activate a set $S_t$ of generalized links at
each time $t$, so that the set of activated links satisfy the following constraints.

\begin{sloppypar}
\begin{list}{}{\itemsep=0pt \parskip=0pt \parsep=0pt \topsep=0pt \leftmargin=0.18in}
\item[1.] {\em Interference constraint (SI):} As we described earlier in this section,
this states that, if $l$ is activated no link in the interference neighborhood $I_l$
can be activated.  The secondary interference constraint along with our notion of
generalized link also captures the following important constraint in multi-band
networks: if a node $v$ is involved in communication with node $u$ over frequency
band $f_j$ over a radio, no other node can communicate with $v$ over $f_j$ at the
same instant.


\item[2.] {\em Maximum Radio Constraint (MR):} Since there are $K$ radios at a
node, at a time, there can be at most $K$ active links that are part of a
node with each link belonging to {\em distinct} bands.

\end{list}
\end{sloppypar}
Our goal in this paper is to derive a distributed scheduling algorithm, that has the following
properties.

{\bf (P1.)} {\em Low Complexity:} The schedule computation complexity should be
independent of the overall network size, number of available bands $M$ and number of
radios $K$. Since this requirement is too stringent, we relax this slightly to allow
the complexity to have poly-logarithmic dependence on the maximum degree but
independent of $M$, $K$, and overall network size.

{\bf (P2.)} {\em Optimal delay scaling:} The scheduling policy should provide an
order-optimal delay guarantee of $E[Q_{tot}]= \bO(H)\ ,$ where $Q_{tot}$ is the total
queue length in the entire network.  Note that the requirement of delay to grow
linearly with number of single-hop pairs is stronger than requiring the delay to be
polynomial in the network size.

{\bf (P3.)} {\em Minimizing throughput-loss:} We want our algorithm to be a
$\mu$-factor throughput optimal (see Definition~\ref{def:to}) algorithm, where $\mu
\geq  1$. This is motivated by the fact that, achieving the delay scaling in P2 with
a polynomial algorithm is bound to result in some loss in
throughput~\cite{shahtradeoff09}.  We will propose MAC such that $\mu$ is
an universal constant under an isotropic propagation model in every frequency band.

\vspace{-0.2in}
\section{Synchronous {\large {\sf MAXIMAL-GAIN}}}
\label{sec:synch}

We now propose our local search based algorithm. In this section, we will assume that
the nodes are synchronized; in a subsequent section, we show how our algorithm can be
adapted to asynchronous setting.

In the rest of this section, we assume a time-slot based system. Nodes are
synchronized and each time-slot has two phases: schedule computation phase and data
transmission phase. Our focus is on the schedule computation phase. The schedule
computation phase is further divided into a certain number of mini-slots over which schedule
is computed. The duration of each mini-slot is just long enough for one round of RTS-CTS message
exchange.

\vspace{-0.1in}
\subsection{Intuition and Algorithm Overview}

A local search based algorithm has the following underlying principle: the solution
iteratively moves from one {\em state} (i.e., feasible solution) to another {\em
neighboring} state such that the new state provides improvement in the objective. The
challenge often lies in appropriate definition of neighborhood-states of every state
(the neighborhood of different states can overlap), such that the algorithm achieves
the desired goals of distributed implementability, fast convergence, and provable
approximation to the optimum. In our case, there is an additional challenge arising
from the fact the weights (queue-lengths) are dynamic quantities. 

To aid our discussion, we first introduce a few notations. First of all, we say
``link $l$ is active in $S$" to mean ``link $l$ is a part of schedule $S$." For
any feasible schedule $S$ (defined by a set of non-interfering  (generalized)
links), we will denote by $S(f_j)$ as the set of active generalized links in
$S$ over frequency band $f_j$.  Clearly, $S=\cup_j S(f_j)$. Also, let $S_{t}$
be a feasible schedule  at time $t$.  Note that, feasibility of $S_{t}$ also
means that no node has more than $K$ (number of radios) active links.  We use
the terms ``state" and ``feasible schedule" interchangeably.  Towards
developing a local search based algorithm, we first define our notion of
neighboring states.

{\bf Neighborhood states:} Let $H_1, H_2, H_3 \hdots$ be a disjoint partition
of all nodes $\mV$ with the following two properties: {\em (i)}  the sub-graph
of $G_1(\mV, \mE_1)$ (the network graph is the lowest frequency band) induced
by $H_k$  has a star-subgraph containing all nodes in $H_k$, i.e., there is a
node $u_k$ (called {\em star-center} or {\em leaders}) that can communicate
with all other nodes in $H_k$, {\em (ii)} the star-centers of $H_1, H_2,
H_3,\hdots$ form an independent set in $G_1(\mV, \mE_1)$.
Figure~\ref{fig:starpart} illustrates this partition. The reason for such a
partition will become clear soon.  We will refer to each $H_k$ as a {\em
group.} For a feasible schedule $S$, let $S^{H_k}$ be the set of {\em active}
links in $S$ outgoing from some node in group $H_k$, i.e., $S^{H_k}=\{l\in
S:\text{tail}(l)\in H_k\}$.  {\em A feasible schedule $S'$ is a neighboring
state of another schedule $S$, if, for each $k$, there is at most one frequency
band where ${S'}^{H_k}$ has additional active links compared to $S^{H_k}$.}
More precisely, $S'$ is a neighbor of $S$, iff, for each $k$, all generalized
links belonging to set ${S'}^{H_k}\setminus S^{H_k}$ operate over the {\em
same} frequency band.  We will denote by $\text{nghbr}(S)$ the neighboring
states of a state $S$.

\begin{figure}[tbh]
\begin{center}
\includegraphics[height=1.5in, width=2.25in]{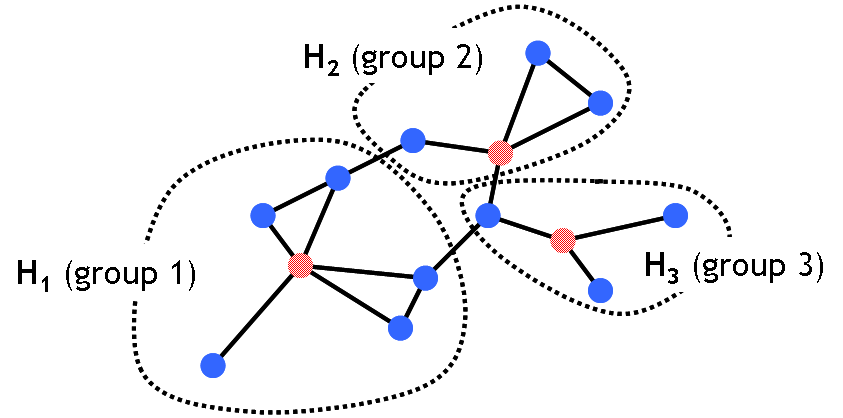}
\caption{{\small Illustration of partitioning the network graph $G_1(\mV,\mE_1)$. 
Each group has a star-subgraph containing all nodes in group; red (shaded) nodes 
are star-centers (leaders).}}
\label{fig:starpart}
\vspace{-0.25in}
\end{center}
\end{figure}

{\bf Overview of the algorithm:} In each time-slot, our MAC chooses a new
schedule by performing a single iteration of local-search. Given a partitioning
of $\mV$ into $H_1, H_2,\hdots$ and our definition of $\text{nghbr}(S)$, the
high-level steps of our MAC algorithm are as follows. 

\begin{enumerate}

\item {\em Random Selection of Bands:} Because of our definition of neighboring
states, for any group $H_k$, a new schedule can
activate additional links in at most one frequency band. In this step, each group independently
chooses a frequency band at random;  this random frequency band is used for
activating new outgoing links from $H_k$. Of course, instead of choosing a band at
random, another option could be to choose a frequency band that yields maximum
improvement in the total weight over the current schedule.  However, such a operation
can incur significant overhead and will have computational and messaging complexity
$\bO(M)$. We show that, random selection of a band just works fine.

\item {\em Computing maximum improvement links:} Once a frequency band is chosen at
random for every group $H_k$, the next step is to find a link/s whose activation over
this band results in maximum improvement in the total weight of all links outgoing
from $H_k$. We will call these links {\em high-priority} links. Note that, if a new
link is activated, some other links may have to be deactivated (removed from the
existing schedule) for the new schedule to be feasible. 
 
\item {\em Computing schedule:} Observe that, the new high priority links cannot be
activated all at once as some of them could interfere with each other. This final
step computes a new feasible schedule by only selecting a non-interfering sub-set
of high-priority links. This selection can be performed using CSMA-CA or any
standard random access protocol (in every band) among the high priority nodes. All
non-high priority links continue transmission {\em only} if no high-priority link
can be sensed.

\end{enumerate}

So far we have just provided an overview of the algorithm, the above steps also have to be performed in a 
distributed manner. In Algorithm~\ref{algo:synchmg}, we show the precise details of how this can be done. 

\vspace{-0.1in}
\begin{remark}[Compelxity]
The complexity of schedule computation is dominated by distributed computation of the
maximum gain link in every group. As we describe in Section~\ref{subsec:lmax}, this can
be performed in $\bO(\ln^2 \Delta)$ mini-lots (the duration of mini-slot is long enough to
exchange an RTS-CTS message) which is essentially the time-complexity of schedule
computation.
\end{remark}
\setcounter{table}{0}
\renewcommand{\tablename}{Algorithm}
\renewcommand{\baselinestretch}{0.95}
{\small
\begin{table}[t]
\hrule
\vspace{0.1in}
\caption{\label{algo:synchmg} \textsc{Algorithm Maximal-Gain Schedule}}
\vspace{-0.1in}
\hrule
{
\begin{algorithmic}[1]

\STATE {\em Grouping and initialization (one time operation at time-$0$):} At
time-$0$, the nodes are partitioned into groups and initialized using the following steps:
\begin{enumerate}
\item 

Use any known distributed computing algorithm (e.g. ~\cite{pargan04}) to
compute an {\em independent dominating set}\footnotemark . 
$D$ in the graph $G_1(V,E_1)$. Nodes in $D$ are called leaders.

\item 
Every node
not in $D$ becomes a {\em follower} of any {\em one} neighboring leader. A leader
along with all its followers are referred to as a {\em group}. Let $H_1, H_2,\hdots$
denote the groups.

\item All nodes in a group choose a common seed for random number generation.

\end{enumerate}

\FORALL {time-slots $t$} 

\STATE Every node uses the seed common to its group to select a frequency band
uniformly at random from the $M$ available bands.

\STATE {\em Gain computation by nodes:} Every node $v$ computes the {\em
gain} in total weight of outgoing links from $v$ if some outgoing link from $v$ were
activated over the randomly selected frequency band. Denoting the randomly selected
band for node $v$ as $f(v)$, this computation is as follows:

\begin{enumerate}

\item For every node $v$, if a new outgoing link from $v$ were to be activated over $f(v)$,
some link $l'$ may have to be deactivated to maintain the feasibility of the new
schedule.  This link $l'$ (with $v$ as one end) is any link that is active over band $f(v)$ in
schedule $S_{t-1}$; if no such link exists in $S_{t-1}$, but all $K$ radios of $v$
are used up in $S_{t-1}$, then $l'$ is the one with the minimum value of $w_lr_l$.
If there is such an $l'$, then $v$ computes $loss(v)=w_{l'}r_{l'}$ based on the
above; else $loss(v)=0$.

\item The net gain due to node $v$ activating a link over band $f_j$ is computed as
follows.
{\small
\begin{align}
\label{eqn:netgain}
&gain(v)= 
\max_{\{l:\text{tail}(l)=v,\text{band}(l)=f_j\}}
\left(w_{l}(t) - loss(v) \right)^+
\end{align}}
\end{enumerate}

\STATE {\em Computing maximum gain link in every group:} Every group computes the outgoing
link in that group with the maximum gain in a distributed manner. 
In Section~\ref{subsec:lmax}, we show that, this can be done with minor modifications 
in CSMA protocol in $\bO(\ln^2 \Delta)$ mini-slots where a mini-slot is the time taken to 
exchange a pair of RTS-CTS message between communicating nodes. At the end of this
step, every node knows which, if any, of its outgoing link has maximum gain. These
nodes are termed {\em high-priority} nodes.

\STATE {\em Contention resolution among high-priority nodes:} Finally, the {\em
high-priority} nodes perform CSMA like contention resolution in their selected
frequency band. All other nodes that
are not high priority, but were active in $S_{t-1}$, continue their transmission
only if they do not sense the channel to be occupied by a high-priority node
(this is inferred by listening to RTS/CTS in the contention resolution phase.)

\ENDFOR
\end{algorithmic}
}
\hrule
\end{table}
}





\vspace{-0.16in}
\subsection{Computing Maximum Gain in a Group}
\label{subsec:lmax}

In the description of {\sf MAXIMAL GAIN} Algorithm, we have assumed that, each group
can compute the maximum gain in the group in $T_{comp}=\bO(\ln^2 \Delta)$ mini-slots. Recall that,
these mini-slots belong to schedule computation phase of a slot.
We now adapt binary search for a broadcast environment to achieve this.

We first argue that, it is sufficient to prove the result by assuming that broadcast
messages from nodes in any group do not collide with that of other groups.  To see
this, suppose we assign color to each group so that any two groups that have
interfering nodes are not assigned the same color\footnote{{\small This is similar to
frequency planning in GSM networks.}}. Under isotropic propagation,
simple geometric considerations show that $\chi=\bO(1)$ colors suffice. Thus, we can
color-code the mini-slots in a round-robin manner and allow each group to perform the
steps of max-computation only in mini-slots corresponding to its color. Thus, if each
group requires $T_{comp}$ mini-slots, the overall computation would require $\chi
T_{comp}=\bO(T_{comp})$ mini-slots. Note that, the coloring is a one-time procedure that can
be done in a distributed manner at the beginning of network operation.

The following simple procedure computes the maximum-gain link in a group in
 $T_{comp}=\bO((\ln \Delta)^2)$ mini-slots.

\HRule
{\sf LOCAL MAX:} Computing maximum gain in a group
\HRule

{\small
Each group is isomorphic to a start graph. We will call the start-center as group leader
and every other node as follower. Suppose $N_v$ is the number of nodes in a group to
which node $v$ belongs. Gain of a node is given by~(\ref{eqn:netgain}).

\begin{sloppypar}
\begin{list}{}{\itemsep=0pt \parskip=0pt \parsep=0pt \topsep=0pt \leftmargin=0.18in}

\item {\em Step~1:} Initially every follower node with positive gain is a candidate winner node.
Initialize broadcast probability of every node $v$ with positive gain as
\[ p_{bc}(v) = \smfrac{1}{2N_v}\]

\item {\em Step~2:} In the next mini-slot, with probability $p_{bc}(v)$, a follower
$v$ sends an RTS message containing its gain (say the value is $g$).  
If there are no collisions, the group leader sends a CTS echoing the gain.

\item {\em Step~3:} In case of CTS, all follower nodes with equal or lower gain than $g$,
remove them from the list of candidate winners.

\item {\em Step~4:} Once every $C_1(1+\ln(N_v))$ mini-slots, every candidate node $v$
updates $p_{bc}(v)$ as $$p_{bc}(v) \leftarrow 1.5 p_{bc}\ .$$ Here $C_1$ is a
positive constant.

\item {\em Step~5:} Step~2\--4 are are repeated with the remaining candidate winners for
$C_2\ln(N_v)(1+\ln(N_v))$ mini-slots, where $C_2$ is a positive constant.

\end{list}
\end{sloppypar}
}
\HRule

\begin{lemma}
\label{LEM:LMAX}
Suppose we choose $C_1=1/\ln(9/8)$ and $C_2=C_1/\ln(1.5)$. Then,
Procedure {\sf LOCAL MAX} computes the maximum gain in a group in $T_{comp}=\bO(\ln
\Delta)^2$ mini-slots with probability at least $1/2$, where $\Delta$ is the maximum degree of the communication graphs over all bands..
\end{lemma}
\begin{proof}
See the longer version~\cite{longer-version12}.
\end{proof}

\vspace{-0.1in}
\subsection{Throughput and Delay Guarantee}

The following result states the throughput and delay guarantee of {\sf MAXIMAL-GAIN}
scheduling algorithm in terms of the graph-parameters $\kappa_1$ and $\beta_{\max}$ defined 
in Section~\ref{sec:setting}. 


\begin{theorem}
\label{THM:TPDE}
The MAXIMAL-GAIN Algorithm has time-complexity of $\bO((\ln \Delta)^2)$ times the time
required to exchange an RTS-CTS message between two neighbors. The algorithms
provides the following guarantees:
\renewcommand{\theenumii}{(\roman{enumii})}
\renewcommand{\labelenumii}{\theenumii}
\begin{sloppypar}
\begin{list}{}{\itemsep=0pt \parskip=0pt \parsep=0pt \topsep=0pt \leftmargin=0.18in}

\item[(i)] MAXIMAL-GAIN algorithm stabilizes all arrival rate vectors $\lambda$ such that
$\lambda \in \frac{\Lambda}{\mu},$
where
\changecolor{
\begin{equation}
\label{eqn:mu}
\mu=\beta_{\max}(1+2(1+\kappa_1))\ .
\end{equation}
}

\item[(ii)] If $\lambda + \epsilon \in \Lambda/\mu$ for some $\epsilon>0$, then the total queue length
achieved by the {\sf MAXIMAL-GAIN} algorithm,
$Q_{tot}(t)=\sum_h q_h(t)$ satisfies
$$ \lim \sup_{T \to \infty} \frac{1}{T} \sum_{t=0}^{T-1} E[Q_{tot}(t)]
= \frac{c_1\ MH}{\epsilon},$$
where $c_1$ is a universal constant independent of the network parameters and
radio resources.
\end{list}
\end{sloppypar}
\renewcommand{\theenumii}{(\arabic{enumii})}
\renewcommand{\labelenumii}{\theenumii}
\end{theorem}

The proof is relegated to the Appendix.

\vspace{-0.15in}

\changecolor{ 
\begin{remark}[On total queue length
bound] Part~(ii) of Theorem~\ref{THM:TPDE} can also be stated in terms of offered load
$\rho<1$. Suppose $\lambda \in \rho \Lambda$. Then under mild technical conditions (see
Proposition~4,~\cite{modiano_delay}) that hold for most practical topologies, one can
choose $\epsilon = c_2(1-\rho/\mu)$ for network
independent constant $c_2$, and thus leading to a total queue-length bound of
$\bO(MH/(1-\rho/\mu))$.  
\end{remark} 
}
\vspace{-0.15in}
\begin{remark}[On the proof of Theorem~\ref{THM:TPDE}]
\label{rem:thmsig}
Since we perform a single iteration of local search in each time-slot, each
subsequent iteration of local search is with a modified weight (queue-length).
The main technical contribution of the proof is in showing that our local
search based algorithm converges even though the weights (or queue-lengths) are
changing dynamically. This is unlike Glauber Dynamics based local search where
the weights must be a slowly-changing function of queue-length for convergence to
happen~\cite{alohaworks}.  
\end{remark}
\vspace{-0.20in}
\changecolor{
\begin{remark}[Refinements of Theorem~\ref{THM:TPDE}]
\label{rem:refine}
The results of Theorem~\ref{THM:TPDE} can be further refined as follows:
\begin{enumerate} 
\item Suppose the MAXIMAL-GAIN algorithm repeats the randomized algorithm for computing LOCAL-MAX 
$\ln(1/\epsilon)/\ln(2)$ times so that the maximal gain in a group is computed with probability at least
$1-\epsilon$ instead. 
Then, the achievable throughput region in Theorem~\ref{THM:TPDE} improves
to $\Lambda/\mu'$ where
\begin{equation}
\label{eqn:mup}
\mu'=\beta_{\max}(1+\smfrac{1+\kappa_1}{1-\epsilon})\ ,
\end{equation}
and time complexity of the algorithm becomes $\bO((\ln \Delta)^2\ln(1/\epsilon))$ times the time
required to exchange an RTS-CTS message between two neighbors.
\item We have stated Theorem~\ref{THM:TPDE} simply in terms of generic graph parameters.
However, we can provide improved bound in terms of one-time {\em grouping} that we
perform. Precisely, let the grouping be such that any generalized link is interfered
by links from at most $\sigma$ distinct groups other than its own group. Then, combined
with $\ln(1/\epsilon)/\ln(2)$ repetitions of randomized LOCAl-MAX, the achievable 
throughput region in Theorem~\ref{THM:TPDE} improves to $\Lambda/\mu''$ where
\begin{equation}
\label{eqn:mupp}
\mu''=\beta_{\max}(1+\smfrac{1+\sigma}{1-\epsilon})\ .
\end{equation}
Thus, if the one-time grouping can be performed efficiently so that $\sigma$ is low, then
the bounds can be much better. 
\end{enumerate}
\end{remark}
}
\vspace{-0.1in}
\changecolor{
\subsubsection{Performance guarantee for special cases} 
Using Theorem~\ref{THM:TPDE} and refinements in Remark~\ref{rem:refine}, we can show 
that he throughput guarantee can be bounded by a constant for many practical networks:
\begin{itemize}
\item {\em Example 1.} Under isotropic propagation, simple geometric considerations 
show that the 2-hop neighborhood MIS cover
is a constant independent of network parameters, \ie, $\kappa_1=\bO(1)$. Along with the
fact that the interference degree, $\beta_{max}$ is a constant under isotropic
propagation,  the throughput loss of {\sf MAXIMAL-GAIN} schedule is at most $\bO(1)$
fraction of full throughput.
\item {\em Example 2.} In a linear chain or ring based network graph, the
throughput guarantee is $1/8$ (using~(\ref{eqn:mupp})) of the maximum throughput. 
Assuming efficient grouping, same bound holds
for a network formed by combining several linear chains where the nodes from two different chains interfere
only where the chains intersect. This kind of network is representative of infrastructure
based wireless networks deployed by placing nodes along city roads.
\item {\em Example 3.} For tree graphs such that the tree junctions are at least 2-hops
away, the throughput guarantee can be shown to be within around $1/8$ of maximum
throughput if the set of group
leaders include all tree junctions. 
\end{itemize}
In Section~\ref{sec:sim}, we show that the actual throughput guarantee of {\sf
MAXIMAL-GAIN} is much better in practice for random topologies and grid. In fact, in all
our evaluations even under asynchronous setting, we found that the throughput loss of our
algorithm is never more than $20\%$.  }

\vspace{-0.1in}
\changecolor{ \subsubsection{Comparison with other algorithms} It is
instructive to compare the theoretical performance of our algorithm with two popular
algorithms: {\em
greedy maximal scheduling} (GMS)~\cite{linshrsri06,shroff} and
Q-CSMA~\cite{srikant_qcsma,alohaworks,jiawal09} across three important
measures: complexity, throughput, and delay. Our algorithm has a throughput
guarantee factor $\mu=\beta_{\max}(1+2(1+\kappa_1))$ that has two
multiplicative terms: a term $\beta_{\max}$ arising due to greedy computation
of maximal gains link within a group, and a term $(1+2(+\kappa_1))$ arising due to
contention resolution among winner links of different groups.  The second term
is essential due to {\em grouping} mechanism which is key to making our
algorithm distributed with a complexity that simply scales as $\bO((\ln
\Delta)^2)$ which is independent of the network size. On the other hand, {\em
greedy maximal scheduling} (GMS)~\cite{linshrsri06,shroff} adapted to our
setting has a throughput guarantee factor of $\beta_{\max}$ and provides
optimal throughput scaling as with our algorithm.  However, the price of GMS is
in terms of much higher computational complexity.  Indeed, even a distributed
implementation of greedy maximal schedule could easily require
$\Omega(MH\log(MH))$ mini-slots~\cite{linshrsri06} and this could be difficult
to implement even for moderate values of number of bands ($M$) and number of
hops ($H$). Thus, the $\kappa_1$ dependent second term in our performance
guarantee can be viewed as the cost of distributed implementability with a
complexity that is practically independent of network size and fully
independent of number of bands. Finally, Q-CSMA or Glauber Dynamics based
algorithm when applied to our setting would provide a throughput of 100\% and
distributed-complexity that is essentially $\bO(1)$; however, the price of such
an algorithm is that total queue length is not guaranteed to be even polynomial
in the network parameters like number of nodes, bands etc. {\em Thus, our
algorithm strikes a trade-off between distributed-complexity, throughput, and
delay guarantee by simultaneously guaranteeing bounded throughput loss, network
independent distributed complexity, and linear (in number of hops) delay
scaling. Our evaluation in Section~\ref{sec:sim} using multi-band asynchronous versions
of the algorithms show that, throughput of MAXIMAL-GAIN is typically
within 80\% of throughput of Q-CSMA and almost on par with GMS, and average delays are
considerably better than both Q-CSMA and GMS.
} }


\vspace{-0.2in}
\subsection{Extensions}
\label{sec:ext}


{\bf Extensions to account for ACI:} Due to leakage in adjacent frequency band caused
by non-ideal filters in transmitters, interference can be also caused by radios
transmitting in adjacent frequency band~\cite{ourdtvpap09}. We can account for this
by defining the interfering links $I_l$ to be a union of links causing secondary
interference and links causing ACI. Some minor modifications are required in the
local computation of gain in {\sf MAXIMAL GAIN} scheduler.  We skip the details.

{\bf Extensions to multi-hop setting:} Our algorithm can be extended to a multi-hop
setting by choosing the weights using ``queue back-pressure" mechanism developed
in~\cite{taseph92}.  Essentially, instead of choosing the weights as queue-lengths, 
they must be chosen as differential queue-length as compared to next hop.
Some minor changes are required to ensure that nodes can learn
suitable information from the neighborhood to compute the queue back-pressure based
weights.

\vspace{-0.15in}
\section{Asynchronous Maximal-Gain}
\label{sec:asynch}

We now use the design philosophy of synchronous Maximal-Gain algorithm to
develop a distributed MAC scheduler for wireless networks where nodes could be
asynchronous or where control packets could be lost. Since the asynchronous model
is difficult to analyze without further assumptions, we
will provide simulation evaluation of this asynchronous MAC in Section~\ref{sec:sim}.

Due to multi-radio multi-band setting, there are two key challenges in extending our synchronous
MAC to an asynchronous setting.  First, since there are fewer radios (per node)
than bands in a typical setting, it is difficult for nodes to keep track of
activities and availabilities of different bands.  Second, nodes in a group may 
not have up to date information on the maximum-gain link.

To alleviate the first problem, we assume that each wireless node has one extra radio
for exchange of control messages on a dedicated control band which is a standard
approach~\cite{vic-ctrl}. All nodes use the control band to send RTS/CTS messages to
reserve an intended band (the intended band is decided using our algorithm described
shortly). To alleviate the second problem, we use the design philosophy of
synchronous Maximal-Gain scheduler as follows.  Recall that in the synchronous
Maximal-Gain algorithm, in each scheduling slot, each group does the following:
{\em(i)} choose a band at random and {\em(ii)} the maximum-gain link
contends for channel access in the chosen band. This can be viewed as assigning priorities to
(generalized) links as follows: any maximum-gain link in a randomly picked band has
{\em highest priority} (HP), the links that transmitted in the previous transmission
opportunity but has no HP link in the interference neighborhood, have {\em medium
priority} (MP), all other links have {\em low priority} (LP) and can transmit only if
no HP and MP links contend. The prioritization of the links can be achieved by
allowing links of different priorities to contend in different parts of contention
window as described shortly.

With the above intuition, we now describe our asynchronous MAC scheduler. In the
following, we assume that the grouping of the nodes is done identical to that
described for the synchronous algorithm. The star-centers are called leaders and they have some 
extra functionalities. Later in this section, we describe a method to eliminate the role of a leader.\\

\HRule The {\bf Asynchronous MAXIMAL-GAIN} Algorithm \HRule
{\small
\noindent
{\bf Leader's algorithm:} The leader in each group performs two key functions:
{\em maintaining, updating, and broadcasting information} about the group members,
and {\em selecting the maximum-gain link} dynamically.

\begin{list}{}{\itemsep=0pt \parskip=0pt \parsep=0pt \topsep=0pt \leftmargin=0.12in}

\item[(1)] {\em Information update:} All nodes send RTS/CTS messages over the
control band for an intended link (the choice of the intended link is described later
in the algorithm). Any RTS/CTS message contains the following:
{\em (i)} queue-length information of the intended link,
{\em (ii)} the frequency band of the intended link,
{\em (iii)} and the data rate of the intended link.
Whenever a leader hears an RTS/CTS message from a
group member over the control band, it does the following:

\begin{itemize}

\item{\em Weight Update:} It updates the queue size of the intended link
from the information embedded in the RTS/CTS messages.

\item{\em Data rate update:} If the data rate of the intended link has changed, then
this is also updated.


\end{itemize}
Nodes also send queue length updates when change in queue length due to packet arrivals is beyond a threshold.

\item[(2)] {\em Maximum-Weight Link Selection:} Once every $T$ time units, a leader
chooses a random band and selects the maximum-gain link in its group using information of the weights and the data rates
of each link (see Eqn.~\ref{eqn:netgain}). The leader then broadcasts over the control
channel the identity of the chosen maximal-gain generalized link, and also the
identity of an existing active link that has to be deactivated (if any). This is because, gain
computation requires that one link is activated at the expense of an existing
active link.  The value of $T$ is a design parameter and is typically of the order of
a few MAC packet transmission times.

\end{list}
\noindent
{\bf Node's algorithm:} For every node, corresponding to {\em each radio} is a set of
3~lists of (generalized) links, namely {\em High Priority (HP)} list, {\em Medium
Priority (MP)} list and {\em Low Priority (LP)} list.  There are two components of
node's algorithm: addition and deletion into these lists, and {\em prioritized}
CSMA/CA based on these lists. These two are described below.

\begin{list}{}{\itemsep=0pt \parskip=0pt \parsep=0pt \topsep=0pt \leftmargin=0.12in}

\item[(1)] {\em Addition and deletion into HP, MP, LP lists:} The lists are updated as
follows for each radio of each node.
\begin{sloppypar}
$\bullet$ {\em HP:} If a node hears a broadcast (over the control band) from its leader
indicating that it is the transmitter of the maximum-gain generalized link, it adds
this link to the HP list. A link is removed from the HP list once it gets activated
for the first time.

$\bullet$ {\em MP:} Any link that gets removed from the HP list is added to the MP list
of the same radio. A link in MP list is removed only when it loses contention of the
channel over the band corresponding to the link.

$\bullet$ {\em LP:} Every other link, that is not a part of HP or MP list belongs to LP
list. It gets removed only when a link is added to any HP list using the method
described above. We ensure that a link belongs to LP list of all radios or none of
the radios.
\end{sloppypar}

\item[(2)] {\em Maximal-Gain CSMA/CA:} Whenever a radio becomes free, it performs the
following steps:
\begin{sloppypar}
\begin{list}{}{\itemsep=0pt \parskip=0pt \parsep=0pt \topsep=0pt \leftmargin=0.12in}

\item[(i)] {\em Step-1:} The radio visits the lists in order of priority (first HP,
followed by MP and then LP) and picks the first inactive link that has non-zero queue
length.

\item [(ii)]{\em Step-2:} The radio switches to the selected link's band and assesses if
the channel is free. If it is free, it transmits an RTS message over the control
channel at a slot selected uniformly between $[0,CW/2-1]$ if it is an HP-link,
between $[CW/2,3CW/4-1]$ if it is an MP-link, and between $[3CW/4, CW-1]$ if it is an
LP-link.  The RTS message contains the following information {\em(i)} the destination
and communication band, and {\em(ii)} queue size of the link.  When the destination
node receives an RTS, it responds with a CTS if resources (radio and band) are
available to it.  The CTS also echoes the information in the RTS.

\item[(iii)] {\em Step-3:} If the contention is successful, then transmission happens
for a constant time interval $TXOPT$ for HP-links, and for one packet duration for
MP and LP links. The constant duration $TXOPT$ is chosen so that multiple packets can
be transmitted (typically a few {\em ms}).

\end{list}
\end{sloppypar}
\end{list}
}
\HRule
\vspace{0.05in}


 {\bf Eliminating the role of group leaders:}  Since grouping is still
necessary, first we ensure that members of the group pick the same random band
at every decision epoch, as required by our algorithm, by using the same source
of randomness (assume small clock drifts). Same source of
randomness within a group can be achieved in the initialization phase by using
the MAC address of the neighboring node who initiates  group formation first.
Second, since the control channel is in the lowest frequency band, all nodes in
a group overhear all RTS/CTS messages thus making sure every node can locally
compute the maximum gain link in the group.

\changecolor{
{\bf Grouping periodicity:} The grouping in our algorithm implicitly assumes a static network. 
In presence of node churn, grouping can be performed {\em incrementally} as follows. New nodes that join can either
join an existing group if there is leader in the neighborhood, else, the node can start a
new group. In addition, a periodic and infrequent {\em re-grouping} can be performed to have
fewer groups.
}

{\bf Overheads:} The RTS/CTS messages in our protocol contain more information
than the standard CSMA/CA RTS/CTS messages. We quantify these overheads
for practical networks. The RTS/CTS message contains three
pieces of information {\em(i)} destination and band, {\em(ii)} queue size of
the link and {\em(iii)} identity of radios to be used at source and
destination.  The first piece of information requires 6 bytes for the
destination MAC and 1 byte for the band (we do not expect more than 256 bands).
For the second piece of information, we assume that the queue size does not
exceed 256 packets, requiring one byte. Finally the identity of the radio
requires not more than 1 byte. Therefore the total overhead is less than 9
bytes.

 {\bf Robustness:} Our protocol is robust in the sense that it works even with
slightly outdated queue size information with the leaders. 
Furthermore, it also does not require the nodes in the group to receive every broadcast 
about maximal-gain links from the leader correctly. We validate this in Section~\ref{sec:sim}.

\vspace{-0.2in}
\section{Evaluation}
\label{sec:sim}

The goals of our simulation experiments are to {\em (i)} investigate average
delay performance of our MAXIMAL-GAIN algorithm {\em (ii)} evaluate the
performance of our algorithm in terms of throughput and {\em (iii)} evaluate
the robustness of our algorithm with respect to loss of control messages in the
network.  \changecolor{We compare asynchronous MAXIMAL-GAIN algorithm against two popular
algorithms: asynchronous and multi-radio multi-band band extension to Q-CSMA
algorithm~\cite{srikant_qcsma} called MB-QSMA, and a multi-radio and multi-band
version of greedy maximal scheduling~\cite{linshrsri06,shroff} called MB-GMS.}
Roughly speaking, Q-CSMA functions as follows. A link is activated if it wins
contention in a CSMA like manner; however, a link decides to contend only if
the following two conditions are met: (a) if no other link in the communication
range of this link was active in the previous slot (b) the link gets a head
upon tossing a coin with probability of head as a suitable increasing  function
of queue length.  MB-QCSMA has two minor changes.  (a) The link contention is
conditional to availability of a free radio at the node and (b) contentions are
done on a control channel with partitions in time, for links of each band.  We
call our extension MB-QCSMA. It is easy to show that MB-QCSMA achieves 100\%
throughput, and thus, this allows us to quantify the throughput loss of asynchronous MAXIMAL-GAIN.
\changecolor{We also perform comparison with MB-GMS where greedy maximal schedule is
computed with generalized-links (instead of links in traditional GMS) every TXOPT
period of time; links that are part of greedy maximal schedule use
RTS-CTS to occupy channel for TXOPT time duration. Recall that high priority links
in asynchronous MAXIMAL-GAIN also stay active for TXOPT time duration.} 

All our algorithms are implemented over a packet-level, realistic, network simulator
called OMNET++~\cite{omnet}. We use the MIXIM framework for 802.11g like physical
layer.

\vspace{-0.1in}
\subsection{Setup}

\textbf{Topology and traffic:} We study the network performance for single-hop
flows for two types of topologies (a)  25 nodes are placed on a grid in a
2500~$m^2$ area. (b) 25 nodes are randomly deployed in the same area, with a
minimum distance of 15~m. Source-destination pairs are generated randomly. To
model burstiness of packet arrivals,  burst length is drawn from a Zipfian
distribution with parameter $1.6$, and burst inter-arrival times are drawn from
an exponential distribution.  To generate realistic interference and RSSI
values, we use standard frequency dependent ITU path loss models~\cite{itu}.

For grid deployment, we present results that are averaged over 60 runs for $5$
seconds (1000 simulation slots)  with different seeds. For random deployments,
we generate 12 random networks and generate traffic with 5 different seeds for
each network.

\textbf{Parameters:} We present the results for nodes with two and three data radios
and using 8-12 frequency bands randomly chosen from among the TV whitespace bands
(512 - 698 MHz). We choose these numbers because we observe that several cities in US
have around these many whitespaces available. Each whitespace band is 6~MHz wide.

\vspace{-0.1in}
\subsection{Results}

\begin{figure}[t]
\begin{center}
\subfigure{
\includegraphics[height=1.5in, width=1.60in]{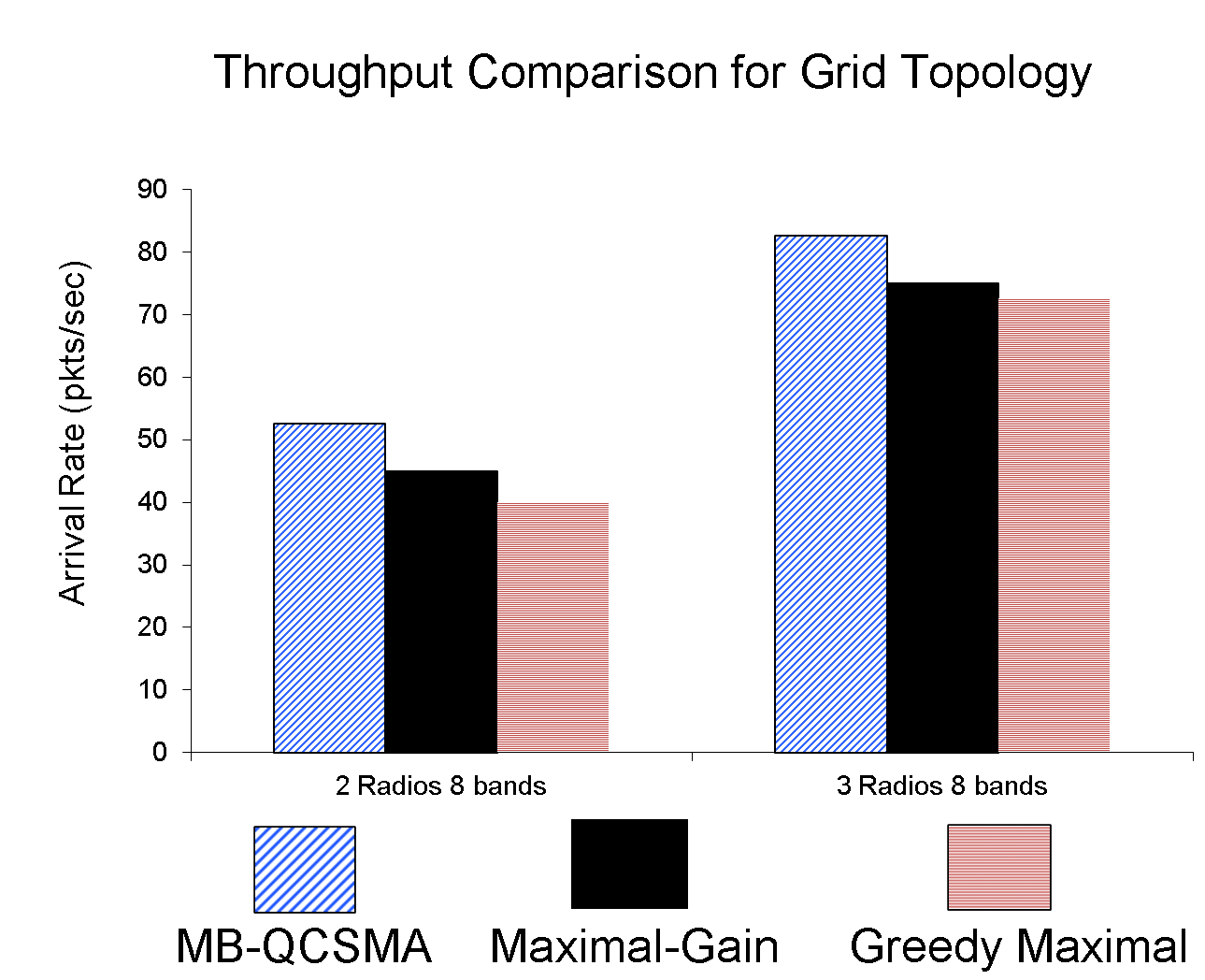}}
\subfigure{
\includegraphics[height=1.5in,width=1.60in]{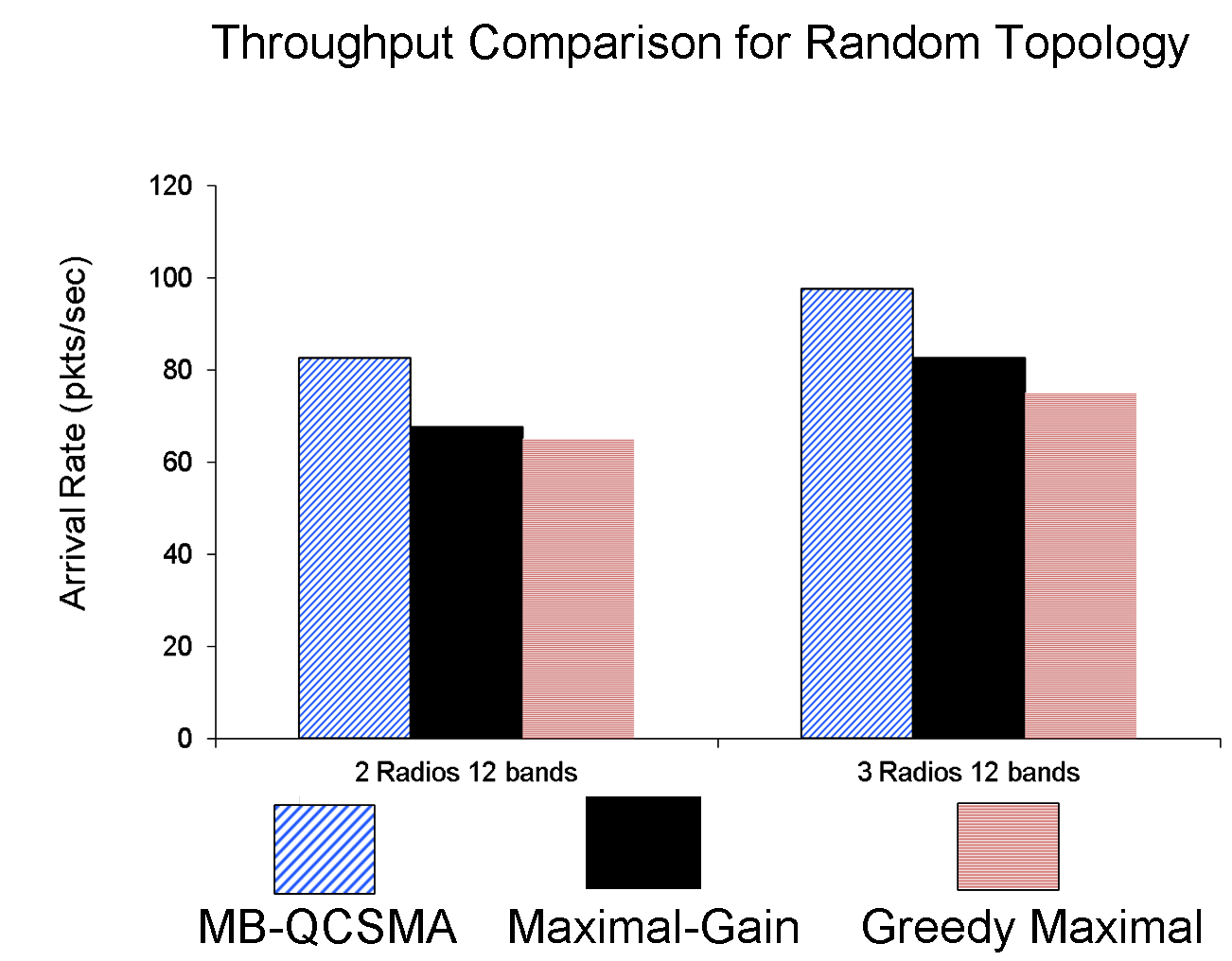}}
\caption{
\label{fig:tptcomp}
Throughput attained by Asyncronous Maximal Gain as compared to MB-QCSMA and MB-GMS. Note that,
ignoring overheads, MB-CSMA achieves 100\% throughput}
\end{center}
\end{figure}

\begin{figure*}[t]
\begin{center}
\subfigure{
\includegraphics[height=1.2in,width=1.7in]{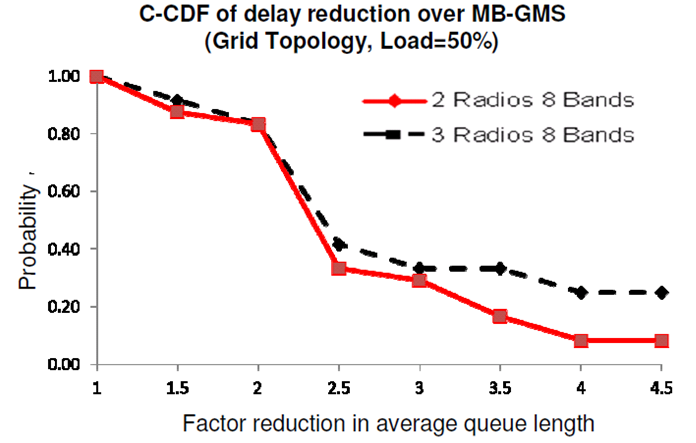}
}
\subfigure{
\includegraphics[height=1.2in,width=1.7in]{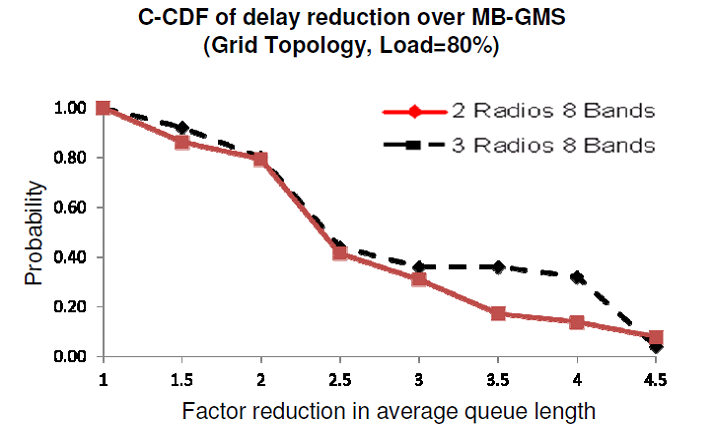}
}
\subfigure{
\includegraphics[height=1.2in,width=1.7in]{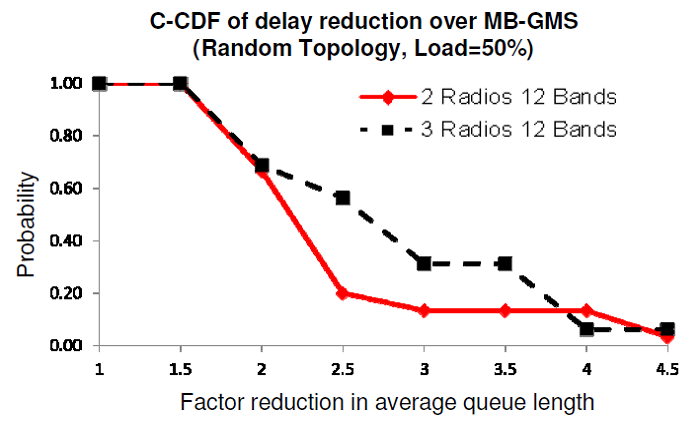}
}
\subfigure{
\includegraphics[height=1.2in,width=1.7in]{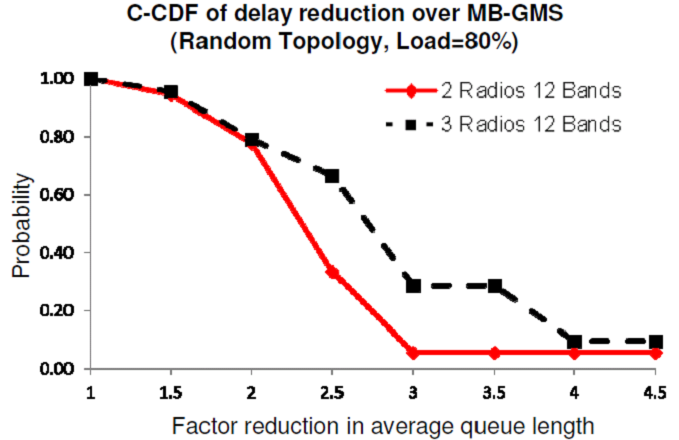}
}
\end{center}
\caption{\label{fig:ccdfgms}
\vspace{-0.2in}
C-CDF of Delay Gains of Asynchronous Maximal Gain as Compared to MB-GMS.}
\end{figure*}
\begin{figure*}[t]
\begin{center}
\subfigure{
\includegraphics[height=1.2in,width=1.70in]{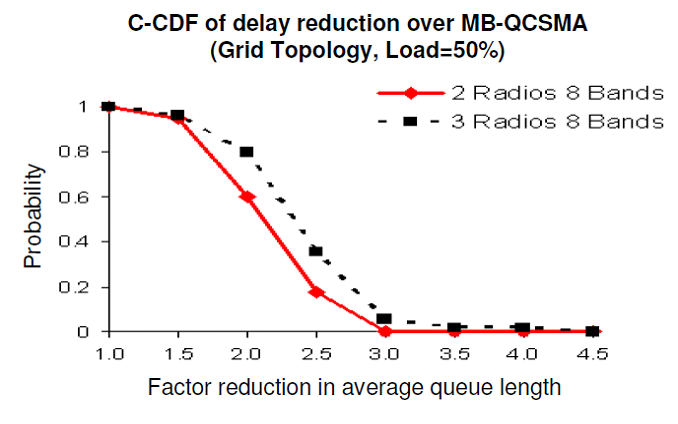}
}
\subfigure{
\includegraphics[height=1.2in,width=1.70in]{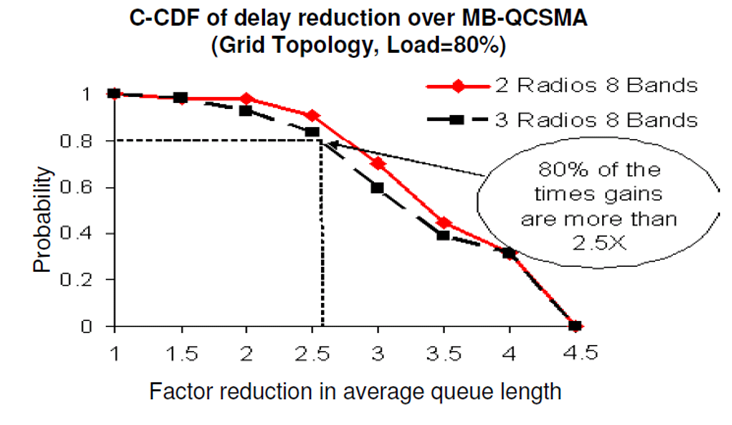}
}
\subfigure{
\includegraphics[height=1.2in,width=1.70in]{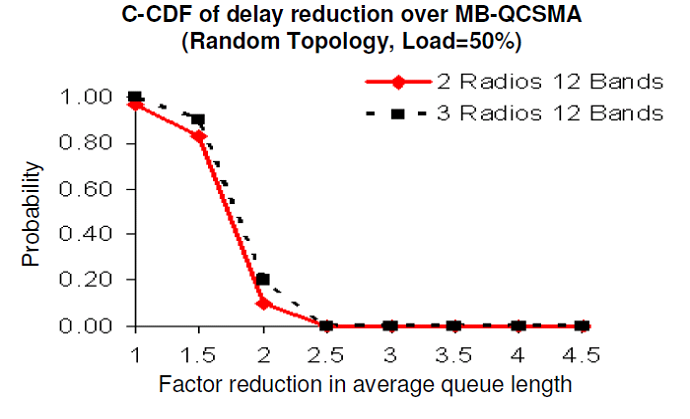}
}
\subfigure{
\includegraphics[height=1.2in,width=1.70in]{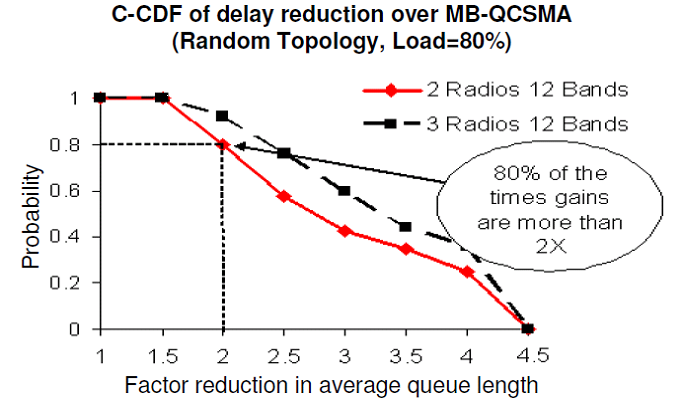}
}
\end{center}
\caption{\label{fig:ccdfqcsma}
\vspace{-0.2in}
C-CDF of Delay Gains of Asynchronous Maximal Gain as Compared to MB-QCSMA.}
\vspace{-0.1in}
\end{figure*}

\begin{figure}[t]
\centering
\subfigure
{\label{fig:q-evol} 
\includegraphics[height=1.2in, width=1.6in]{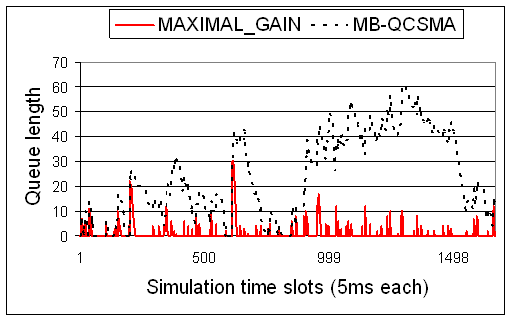}}
\subfigure
{\label{fig:q-cdf}
\includegraphics[height=1.2in, width=1.6in]{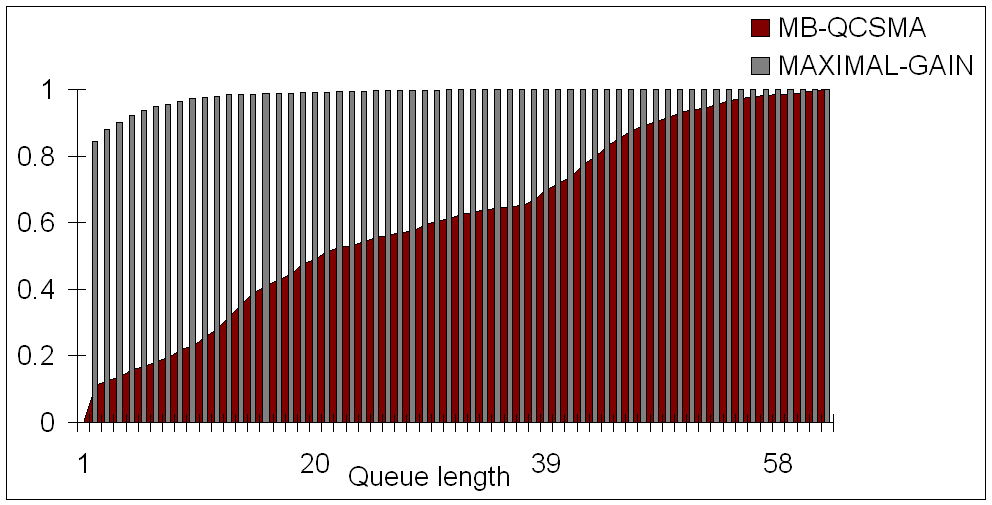}}
\vspace{-0.1in}
\caption{{\bf (a):} Queue evolution of a typical bottleneck node over 10~seconds. {\bf (b):}
CDF of the queue length values of a bottleneck node over 10 seconds. Note that the
probability of queue length being less than 25 is close to 50\% with MB-QCSMA while
it is close to 100\% in MAXIMAL-GAIN.}
\vspace{-0.2in}
\end{figure}

\textbf{Delay gains:} We compare the network-wide average delay of MAXIMAL-GAIN with
MB-GMS (Figure~\ref{fig:ccdfgms}) and MB-QCSMA (Figure~\ref{fig:ccdfqcsma}).  We show the
C-CDF of delay improvements (average queue length with MB-QCSMA / Average queue length
with MAXIMAL-GAIN) with our asynchronous algorithm for a moderate load (50\%) and high
load (80\%).  to answer the question : how often are the delay gains significant?  We show
the results for 25 node grid deployment with 8 bands and 25 node random deployment with 12
bands.  \changecolor{Compared to MB-QCSMA, as shown in Figure~\ref{fig:ccdfqcsma}, under medium load
conditions, the gains are $2\times$ ($1.75\times$) or more in 50\% of the runs, and
$1.75\times$ ($1.5\times$) or more in 80\% of the runs for grid (random) topology. Also,
compared to MB-GMS, as shown in Figure~\ref{fig:ccdfgms}, under high-load conditions the
gains are $2\times$ ($2.25\times$) or more in 50\% of the runs, and $1.75\times$
($2\times$) or more in 80\% of the runs for grid (random) topology.  Similar gains were
observed with 12~(8) bands for grid~(random) topologies.}

To understand this network-level behaviour more clearly, we observe the behavior of
queues at individual nodes in the network.  Figures~\ref{fig:q-evol} and
~\ref{fig:q-cdf} plot the queue evolution of one of the congested nodes in the
network over a period of 10~seconds. We clearly see that the MB-QCSMA algorithm
allows the queue size to grow large before giving the link preferential treatment and
draining it.  On the other hand, the MAXIMAL-GAIN algorithm drains all queues more
uniformly. In Figure~\ref{fig:q-cdf}, we plot the CDF of the queue sizes and
we observe that, with MAXIMAL-GAIN, queue length is less than 10 for 97\% of the time, 
whereas queue length is less than 10 in only 15\% of the time with MB-QCSMA. We also observed
similar trends in comparison to MB-GMS.

Our key takeaways are as follows:
\begin{itemize}
\item The delay gains are more significant (typically in the range
$2\--4\times$) under heavy load.  

\item The gains are more prominent for regular topologies like grid, pointing towards the
increased benefit of our MAC in scenarios of careful network deployment.

\item Our algorithm not only shows gains in average queue-length or delay, but also
ensures short queue-lengths on most times during network operation. Thus, our MAC
allows the designer to use a small MAC-layer buffer.

\end{itemize}

\textbf{Throughput loss:} We increase the average arrival rates for flows in the network,
till the queues at the nodes start exploding. We call the maximum average arrival rate
that the network can support as the maximum stabilizable rate of the network under a given
scheduling algorithm.  As shown in Figure~\ref{fig:tptcomp}, the maximum stabilizable rate
of MAXIMAL-GAIN algorithm is within 80-90\% of the maximum stabilizable rate of Q-CSMA
algorithm.  Also, in most instances, the throughput attained by our algorithm is on par
with or better than MB-GMS. Thus, we summarize that, {\em the throughput loss of our delay
oriented MAXIMAL-GAIN MAC is no more than $10\--20\%$}.

\textbf{Robustness:} MAXIMAL-GAIN algorithm requires the leader to overhear all
RTS/CTS communications in the group for correct operation. In practice, it is
possible that the leader does not hear all RTS/CTS messages, due to collisions and
losses and hence has some outdated information. To measure the robustness of our
algorithm to such gaps in accuracy, we vary the probability with which the RTS/CTS
messages are overheard by the leader and measure the resulting throughput and delay in the network.
We did not observe any significant drop in
either throughput or delay performance across several runs even when the probability
of not hearing messages was increased to 50\%. Thus, our algorithm is robust to 
somewhat outdated information.

\vspace{-0.15in}
\section{Conclusion}
\label{sec:conc}

In this paper we have shown that design of practical distributed algorithms for
multi-band multi-radio networks are possible that achive low delays, almost
achieve full throughput, and have low overhead and complexity.

{\scriptsize

}

\appendices

\section{Proof of Theorem~\ref{THM:TPDE}}
\label{app:thmproof}

We will start by proving Part~1 of Theorem~\ref{THM:TPDE}. Then, we will show that
the Part~2 of Theorem~\ref{THM:TPDE} follows from Part~1 along the lines of proof of
delay bound in~\cite{modiano_delay}.

{\bf Proof of Part~1 of Theorem~\ref{THM:TPDE}:}
The max-weight ($\mw$) schedule, defined by the schedule that maximizes the total weight
(queue length times the link rate), in every time-slot stabilizes the
network~\cite{taseph92}. In~\cite{giaprasha03}, the result was generalized to show
that, any schedule such that the total weight differs from the max-weight schedule at
most by an additive constant also stabilizes the network. While~\cite{giaprasha03}
showed it for switch scheduling, the technique is very general and can be adapted to
show the following generic result for our framework.

\begin{lemma}[Adapted from~\cite{giaprasha03}]
\label{lem:approxmw}
Suppose a randomized algorithm has total weight $W(t)$ at time $t$ and let $W^*(t)$
be the weight of the $\mw$ schedule.
If
\begin{equation}
\label{eqn:approxmw}
\mathbb{E}\left[W^*(t) - k W(t)\right] < C\ ,
\end{equation}
for some $C<\infty$ and $k<\infty$, then the algorithm
stabilizes any arrival $\bf{\lambda} \in \frac{\Lambda}{k}$.
\end{lemma}

Thus, if we show that {\sf MAXIMAL-GAIN} schedule
satisfies~(\ref{eqn:approxmw}) for $k=\mu$ as defined in Theorem~\ref{THM:TPDE},
then Part~1 immediately follows. We will now argue that the condition given
by~(\ref{eqn:approxmw}) is satisfied by a scheduling algorithm, if the following inequality holds for
some \changecolor{ $\card{\alpha}<1$:}
\begin{equation}
\label{eqn:expdec}
\mathbb{E}[W^*(t) - \mu W(t)] \le \alpha \mathbb{E}[W^*(t) - \mu W_{-1}(t)]\ ,
\end{equation}
where $W_{-1}(t)$ is the total weight at time $t$ of the schedule at time $t-1$
(denoted by $S_{t-1}$). In other words, if we activate at time $t$ the same set of
links as the schedule at $t-1$, the total weight is denoted by $W_{-1}(t)$. We will
now show~(\ref{eqn:expdec})$\Rightarrow$(\ref{eqn:approxmw}).
 
\changecolor{First consider the case $0\leq \alpha < 1$.} We introduce some notations. Let $\delta_1$ be the
upper bound on the expected total weight increase that is possible from
one slot to the subsequent slot if the arrival rates are in the stable region. Similarly, let
$\delta_2$ be the bound on maximum decrease in weight possible from one slot to the
subsequent slot. Clearly, there exists $\delta_1<\infty$ because
{\small
\begin{align*}
\Ex[W^*(t)\ |\ W^*(t-1)] & \leq
W^*(t-1) +  \sum_{h\in \mh}\Ex[A_h(t)] \\
& = W^*(t-1)+\sum_h\lambda_h\\
\Rightarrow  \Ex[W^*(t)] & \leq \Ex[W^*(t-1)] + \sum_h\lambda_h
\end{align*}
}
which follows from the fact that, from $t-1$ to $t$, each weight increases by
at most the number of arrivals at $t$. Thus, we can choose $\delta_1=\sum_h
\lambda_h$.  Also, $\delta_2 < \infty$ because the number of packets that can
be served over a time slot is bounded. Now, using the facts that
$\Ex[W^*(t)]\leq \Ex[W^*(t-1)]+\delta_1$ and $W_{-1}(t)\geq W(t-1)-\delta_2$
(because  $W_{-1}(t)$ and $W(t-1)$ are the total weight of the same schedule but
in successive time-slots), we have the following if~(\ref{eqn:expdec}) holds:
{\small
\begin{align}
& \mathbb{E}[W^*(t) - \mu W(t)] \\ \nonumber
& \le  \alpha \mathbb{E}[W^*(t) - \mu W_{-1}(t)] \\ \nonumber
& \le  \alpha (\mathbb{E}[W^*(t-1)] - \mu \mathbb{E}[W(t-1)])  + 
 \alpha(\delta_1 + \mu \delta_2)\\ \nonumber
& \le  \alpha^2 (\mathbb{E}[W^*(t-2)] - \mu \mathbb{E}[W(t-2)]) +
 (\delta_1 + \mu \delta_2)(\alpha + \alpha^2) \\ \nonumber
& \le  \frac{\alpha}{1-\alpha}((\delta_1 + \mu \delta_2))\ ,
\label{eqn:diffmaxw}
\end{align}}
where, we have assumed that the system was started at time $t= -\infty$ with $W^*=0$ and $W=0$.

\changecolor{
The case when $-1<\alpha \leq 0$ follows by defining $\delta_1$ as the upper bound
on the {\em decrease} of expected total weight that is possible from
one slot to the subsequent slot, and by defining
$\delta_2$ be the bound on maximum {\em increase} in weight possible from one slot to the
subsequent slot. A similar calculation can then be performed.
}

Thus, we have argued that,~(\ref{eqn:expdec}) implies~(\ref{eqn:approxmw}) with
$k=\mu$ which in turn implies Part~1 of Theorem~\ref{THM:TPDE}.
The following Lemma asserts that Equation~\ref{eqn:expdec} holds from which the
desired result follows.

\begin{lemma}
\label{lem:expdec}
The Maximal-Gain algorithm satisfies
\begin{equation*}
\mathbb{E}\left[W^*(t) - \mu W(t)\right] \le \alpha \mathbb{E} \left[W^*(t) - \mu
W_{-1}(t)\right]\ ,
\end{equation*}
where
\changecolor{
\[\mu=\beta_{\max}(1+2(1+\kappa_1))\ ,\ 
\alpha=1-(\theta + \smfrac{\theta}{2(1+\kappa_1)})\ ,\]
and $\theta=1-(1-\smfrac{1}{M})^{1+\kappa_1}$.
}
\end{lemma}

\vspace{-0.15in}
\changecolor{
\begin{remark}
Note that $\card{\alpha}<1$ because $0<\theta + {\theta}/{(2(1+\kappa_1))}<2$. This is
true for any $M$ and $\kappa_1$.
\end{remark}
}
\begin{proof}

Let $S_t$ be the set of activated/scheduled links in time slot $t$ and let $S_t(f_j)$
be the set of those links that operate over frequency band $f_j$.  We will index by
$k$ the groups formed by the grouping in {\sf MAXIMAL-GAIN} scheduling.  Let
\changecolor{$L_k(j,t)\subseteq S_t(f_j)$} be the set of all 
generalized links on band $f_j$ \changecolor{activated at time $t$} such that the node
tail($l$) belongs to group $k$.

First, we lower bound $W(t) - W_{-1}(t)$.  To do this, we
consider two kinds of links: {\em winner links} that are egress from a node
\changecolor{based on the outcome of Algorithm~LOCAL-MAX described in
Section~\ref{subsec:lmax}}, and every other links that were also active in the
previous time slot but who lose their weight as new winner links get activated.
Clearly, the winner nodes that become a part of the schedule deactivate some links in
their interference neighborhood.  Let $W^+(t)$ be the random variable denoting the
total gain in weight of winner links at time $t$ compared to their total weight if
$S_{t-1}$ was also used at time $t$; and let $W^-(t)$ the random variable denoting
the loss in weight of deactivated links compared to their total weight if $S_{t-1}$
was also used at time $t$.
\begin{align}
&W(t) - W_{-1}(t)=W^+(t)- W^-(t)
\label{eqn:gainloss}
\end{align}
In the following, we will first lower bound \changecolor{$\Ex[W^+(t)|S_t,\mathbf{w}(t)]$ and then upper bound
$\Ex[W^-(t)|S_t,\mathbf{w}(t)]$}.
We will now introduce some notations.
Recall that in the {\sf MAXIMAL-GAIN} algorithm, each group leader selects a band at random
and picks the links with the maximum gain to transmit. This link then does CSMA
contention to transmit on the chosen band. It may well be that the chosen link is
unable to transmit since some other link in an interfering group wins the contention
on that band.  Define the indicator variable $I_{k,j}$ as
{\small
\[I_{k,j}=
\left\{
\begin{array}{ll}
1 & \text{if group $k$ chooses band $f_j$ at random}\\
  & \text{and winner captures channel}\\
  & \\
0 & \text{else}
\end{array}
\right.\]
We also define
\[I^{\max}_{k,j}=
\left\{
\begin{array}{ll}
1 & \text{if group $k$ succeeds in computing}\\
  & \text{ the max gain in the group}\\
  & \\
0 & \text{else}
\end{array}
\right.\]}
Now define $d_l(t)$ and $e_l(t)$ as the total weight decrease due to deactivation of other active links
(that were active in $S_{t-1}$) at tail($l$) and head($l$), respectively.
Now note that, for all possible sample paths, we can bound the random variable
$W^+(t)$ as
{\small \begin{align}
& W^+(t)
\geq \sum_{k,j} I_{k,j} I^{\max}_{k,j}\max_{l\in L_k(j,t\changecolor{-1})}
( ([w_l(t) r_l  - d_l(t)])^+)
\label{eqn:gain}
\end{align}}
because the term $\max_{l\in L_k(\changecolor{f_j,t-1})}[w_l(t) r_l  - d_l(t)]^+$ represents the
maximum gain in group $k$ at time $t$ (see \changecolor{eqn.~(\ref{eqn:netgain}) in
Step-4(2) of Algorithm~\ref{algo:synchmg}}).
We also have,
{\small \begin{eqnarray*}
& & \max_{l\in L_k(j,t\changecolor{-1})}[w_l(t) r_l - d_l(t)]^+ \\
&\geq&  \max_{l \in L_k(j,t\changecolor{-1}) \cap L^*(j,t)} [w_l(t) r_l - d_l(t)]^+\\
&\geq& \max_{l \in L_k(j,t\changecolor{-1}) \cap L^*(j)} [w_l(t) r_l] -
\sum_{l\in L_k(j,t\changecolor{-1})}d_l(t)
\end{eqnarray*}}
Here \changecolor{$L^*(j,t)$ is the set links of which are allocated band $j$ in the optimal
allocation at time $t$}. The first step follows from the fact constraining the allocation over the
set $L_k(j,t) \cap L^*(j,t)$ will only decrease the value of the expression, and the second
step follows from the fact that maximum of values of elements is less than the sum of their
values. Note that  we have,
{\small
$$\changecolor{\max_{l \in L_k(j,t-1) \cap L^*(j,t)} [w_l(t) r_l] \geq
\frac{1}{\beta_{\max}} \sum_{l \in L_k(j,t-1) \cap L^*({j},t)}{w_l(t)r_l}},$$
}
where $\beta_{\max}$ is the
interference degree of the network. This is because $\max$ is greater than the
average and there can be at most $\beta_{\max}$ links active on a band in a group.
We thus have
{\small
\begin{align}
& \max_{l\in L_k(j,t\changecolor{-1})}[w_l(t) r_l - d_l(t)]^+ \nonumber \\
& \geq \left( \sum_{l \in L_k(j,t\changecolor{-1}) \cap L^*(j,t)} \frac{\changecolor{w_l(t)r_l}}{\beta_{\max}}
- \sum_{l\in L_k(j,t\changecolor{-1})}d_l(t) \right)^+
\label{eqn:gain1}
\end{align}}
Substituting~(\ref{eqn:gain1}) in~(\ref{eqn:gain}) followed by taking expectations
in~(\ref{eqn:gain}), we obtain
{\small\begin{align}
& \Ex\left[W^+(t)|S_{t-1}, \mathbf{w}(t)\right]
\label{eqn:gain2}\\
& \geq \sum_{k,j} \smfrac{1}{2}p_{k,j}\left(
\sum_{l \in L_k(j,t\changecolor{-1}) \cap L^*\changecolor{(j,t)}} \frac{\changecolor{w_l(t)\changecolor{r_l}}}{\beta_{\max}}
- \sum_{l\in L_k(j,t)}d_l(t) \right)^+\ ,
\nonumber\\
\end{align}}
where $\Pr\{I_{k,j}=1\}=p_{k,j}$ and we have also used the fact that
$\Pr\{I^{\max}_{k,j}=1\} \geq 1/2$ (from Lemma~\ref{LEM:LMAX}).
Denoting by $X_{k,j}$ the random variable for the number of winner links of other groups
on band $j$ that interfere with group $k$'s winner at time $t$ (these links are also picked as winners by their
respective leaders over the same band as $k$). Assuming that the contention resolution is perfect
in the sense that (i) with probability one, at least one of the contending nodes grab channel access, and
(ii) all contending nodes are equally likely to get channel access,
we can now bound $p_{k,j}$ as follows:
\changecolor{
\begin{align*}
& p_{k,j}=\Pr(\text{leader chooses band $f_j$}) \\
&\ \ \times
\Pr\left(
\begin{array}{c}
\text{winner link in group}\\
\text{wins contention}
\end{array}
\ |\ \text{band $f_j$ is chosen}\right)\\
&=\frac{1}{M}\Ex \left[\frac{1}{1+X_{k,j}}\right]
\end{align*}
Since the number of interfering links of winner link in a group is upper bounded by
$\kappa_1$, $X_{k,j}$ can shown to be stochastically dominated from above by the
distribution $\textsc{Bin}(\kappa_1,1/M)$ (i.e., $\Pr(X_{k,j}\geq n)\leq
\Pr(\textsc{Bin}(\kappa_1,1/M)\geq n)\ \forall n\geq 0$). Thus it follows that, 
\begin{align}
p_{k,j} & \geq \frac{1}{M} \Ex\left[ \frac{1}{1+ \textsc{Bin}(\kappa_1,\smfrac{1}{M})}\right]
\label{eqn:pkj} \\
& = \frac{1}{M}\sum_{n=0}^{\kappa_1}\smfrac{1}{1+n}{\kappa_1 \choose n}
\left( \smfrac{1}{M} \right)^n \left(1- \smfrac{1}{M}\right)^{\kappa_1 - n}
=\smfrac{\theta }{1+\kappa_1}\ , 
\nonumber
\end{align}
where $\theta = 1 - (1-1/M)^{\kappa_1+1}$ and the last equality follows
from standard computations with Binomial distribution.
}
Substituting~(\ref{eqn:pkj}) into~(\ref{eqn:gain2}) we obtain the 
following.
\begin{align}
& \Ex\left[W^+(t)|S_{t-1}, {\mathbf w}(t)\right]
\nonumber\\
& \geq \frac{\theta}{2(1+\kappa_1)}\sum_{k,j} \left(
\changecolor{\sum_{l \in L_k(j,t) \cap L^*(j,t-1)} \smfrac{w_l(t){r_l}}{\beta_{\max}}}\right.
\nonumber \\
& - \left. \sum_{l\in L_k(j,t\changecolor{-1})}d_l(t) \right)
 \geq \frac{\theta}{2(1+\kappa_1)}
\left(\smfrac{W^*\changecolor{(t)}}{\beta_{\max}} \changecolor{-} W_{-1}(t)\right)
\label{eqn:gain3}
\end{align}

We will now derive an upper bound on $W^-(t)$. Note that,

\changecolor{
\begin{align}
&{\Ex}\left[W^-(t)\ |\ S_{t-1}, \mathbf{w}(t)\right]
\nonumber \\
&=\sum_{l\in S_{t-1}} w_l(t) r_l. 
\Pr\left(\begin{array}{c}
l\ \text{is in the interference neighborhood}\\
\text{of a winner link}
\end{array}\right) 
\nonumber \\
& \leq \sum_{l\in S_{t-1}}w_l(t) r_l (1-(1-\smfrac{1}{M})^{1+\kappa_1}) = \theta W_{-1}(t)
\label{eqn:loss2}
\end{align}}

\changecolor{
Substituting,~(\ref{eqn:gain3}) and~(\ref{eqn:loss2}) into~(\ref{eqn:gainloss}), we
finally get
{\small \begin{equation}
\Ex[W(t)-W_{-1}(t)]
\geq W^*(t)\frac{\theta}{2\beta_{\max}(1+\kappa_1)} - 
(\theta + \smfrac{\theta}{2(1+\kappa_1)}) W_{-1}(t)
\end{equation}}
which upon further rearrangement of term yields
{\small \begin{align*}
\mathbb{E} \left[W^*(t) - \mu W(t)\right]
\leq \alpha \mathbb{E} \left[W^*(t) - \mu W_{-1}(t)\right]
\end{align*}}
where
\[\mu=\beta_{\max}(1+2(1+\kappa_1))\ ,\ 
\alpha=1-(\theta + \smfrac{\theta}{2(1+\kappa_1)})\ ,\]
where $\theta=1-(1-\smfrac{1}{M})^{1+\kappa_1}$.
}
\end{proof}

As argued before, Lemma~\ref{lem:expdec} implies~Lemma~\ref{lem:approxmw} with
$k=\mu$ which in turn implies Part~1 of Theorem~\ref{THM:TPDE}.

{\bf Proof of Part~2 of Theorem~\ref{THM:TPDE}:} This part follows along the
lines of the delay analysis in ~\cite{modiano_delay} with some minor
modifications. See~\cite{longer-version12} for a proof sketch.
\hfill $\blacksquare$

\end{document}